\documentclass{article}
\usepackage{cite}
\usepackage{microtype}
\usepackage{graphicx}
\usepackage{xspace}
\usepackage{amsmath}
\usepackage{amssymb}
\usepackage{amsthm}
\usepackage{url}
\usepackage{fullpage, times}
\usepackage{amsfonts,color,float,verbatim}
\usepackage{amssymb, url, bbm}

\usepackage{algorithm, algorithmic}

\newcommand{\xn}{X^{n_A}}
\newcommand{\xnstar}{X^{n_{A^*}}}
\newcommand{\xnml}{X^{n_A}_{-\ell}}
\newcommand{\xnmlstar}{X^{n_{A^*}}_{-\ell}}

\newcommand{\xnmltwo}{X^{n_A}_{-\ell_1, -\ell_2}}

\newcommand{\txn}{\mathrm{T_\ell}[X^{n_A}_{-\ell}](x_\ell)}
\newcommand{\txntwo}{\mathrm{T_{\ell_1, \ell_2}}[X^{n_A}_{-\ell_1, -\ell_2}](x_{\ell_1}, x_{\ell_2})}

\newcommand{\tuxn}{\mathrm{T^U_\ell}[X^{n_U}_{-\ell}](x_\ell)}
\newcommand{\tjxn}{\mathrm{T^{(j)}_\ell}[X^{n_{A_j}}_{-\ell}](x_\ell)}

\newcommand{\xhatfull}{F_\ell[X^{n_A}_{-\ell}]}

\newcommand{\xnj}{X^{n_{A_j}}}
\newcommand{\xnjml}{X^{n_{A_j}}_{-\ell}}

\newcommand{\xon}{X^{n_{U}}}

\newcommand{\xnu}{X^{n_U}}
\newcommand{\xnuml}{X^{n_{U}}_{-\ell}}

\newcommand{\atuxn}{\mathrm{T^U_\ell}[X^{n_U}_{-\ell}](x_\ell)}
\newcommand{\atjxn}{\mathrm{T^{(j)}_\ell}[X^{n_{A_j}}_{-\ell}](x_\ell)}
\newcommand{\afuxn}{\mathrm{F^U_\ell}[X^{n_U}_{-\ell}]}
\newcommand{\afjxn}{\mathrm{F^j_\ell}[X^{n_{A_j}}_{-\ell}]}

\newcommand{\ept}{\operatornamewithlimits{E}}

\newcommand{\eat}[1]{}
\newcommand{\edit}[1]{{#1}}
\newcommand{\kedit}[1]{{#1}}

\newenvironment{tighterdescription}%
  {\begin{description}%
    \setlength{\itemsep}{1pt}%
    \setlength{\parskip}{1pt}}%
  {\end{description}}

  {\begin{enumerate}%
    \setlength{\itemsep}{1pt}%
    \setlength{\parskip}{1pt}}%
  {\end{enumerate}}

  {\begin{itemize}%
    \setlength{\itemsep}{1pt}%
    \setlength{\parskip}{1pt}}%
  {\end{itemize}}

\newtheorem{theorem}{Theorem}[section]
\newtheorem{lemma}[theorem]{Lemma}

\newtheorem{definition}[theorem]{Definition}
\newtheorem{observation}[theorem]{Observation}

\newtheorem{condition}[theorem]{Condition}

\newcommand{\romTS}{\hat{n}}

\newcommand{\romHLL}{\mathrm{HLL}}

\newcommand{\pIKMV}{\mathrm{pKMV}}
\newcommand{\romKMV}{\mathrm{KMV}}
\newcommand{\romAdapt}{\mathrm{Adapt}}
\newcommand{\romIAdapt}{\mathrm{multiAdapt}}

\newcommand{\romIKMV}{\mathrm{multiKMV}}
\newcommand{\IKMV}{multiKMV}

\newcommand{\rommultiKMV}{\textrm{multiKMV}}

\newcommand{\estonsub}{\mathrm{EstimateOnSubPopulation}}
\newcommand{\thetaunion}{\mathrm{ThetaUnion}}

\newcommand{\samp}{\mathrm{samp}}

\newlength\myindent
\newcommand\bindent{%
  \begingroup
  \setlength{\itemindent}{\myindent}
  \addtolength{\algorithmicindent}{\myindent}
}
\newcommand\eindent{\endgroup}

\usepackage{acronym}
\newcommand{\distinctsub}{\textsc{DistinctOnSubPopulation}}
\newcommand{\distinctsubP}{\textsc{DistinctOnSubPopulation}$_P$}
\newcommand{\distinctP}{\textsc{Distinct}$_P$}

\newcommand{\distinct}{\textsc{Distinct}}
\newcommand{\distinctA}{\textsc{Distinct}$(A)$}
\newcommand{\DistinctElements}{\textsc{DistinctElements}}

\title{A Framework for Estimating Stream Expression Cardinalities} 

\author{Anirban Dasgupta\thanks{Indian Institute of Technology, Gandhinagar} \and Kevin Lang\thanks{Yahoo Labs} \and Lee Rhodes\thanks{Yahoo! Inc} \and Justin Thaler\thanks{Yahoo Labs}}

\date{}

\begin{document}

\maketitle

\begin{abstract}
Given $m$ distributed data streams $A_1, \dots, A_m$, we consider the problem 
of estimating the number of unique identifiers in streams defined by set expressions over $A_1, \dots, A_m$.
We identify a broad class of algorithms for solving this problem, and show that the estimators output by any algorithm in this class 
are perfectly unbiased and satisfy strong variance bounds. Our analysis unifies and generalizes a variety of earlier results in the literature. 
To demonstrate its generality, we describe several novel sampling algorithms in our class, and show that they
achieve a novel tradeoff between accuracy, space usage, update speed, and applicability. 
\end{abstract}

\section{Introduction} 
Consider an internet company that monitors the traffic flowing over its network by placing 
a sensor at each ingress and egress point. Because the volume of traffic is large,
each sensor stores only a small \emph{sample} of the observed traffic, using some simple sampling procedure. 
At some later point, the company
decides that it wishes to estimate the number of unique users who satisfy a certain property $P$ and have communicated over its network.
We refer to this as the \distinctsubP\ problem, or \distinct$_P$ for short.
How can the company combine the samples computed by each sensor, in order to accurately estimate the answer to this query?
 
In the case that $P$ is the trivial property that is satisfied by all users, the answer to the query is simply the 
number of \DistinctElements\ in the traffic stream, or \distinct\ for short. The problem of designing streaming algorithms and sampling procedures
for estimating \DistinctElements\ has been the subject of intense study. In general, however, $P$ may be significantly more complicated than the trivial property, and may not be known until query time. 
For example, the company may want to estimate the number of (unique) men in a certain age range,
 from a specified country, who accessed a certain set of websites during a designated time period, while excluding IP addresses belonging to a 
 designated blacklist.  
This more general setting, where $P$ is a nontrivial ad hoc property, has received somewhat less attention than the basic \distinct\ problem.

In this paper, our goal is to identify a simple method for combining the samples from each sensor, so that the following holds.
As long as each sensor is using a sampling procedure that satisfies a certain mild technical condition, then for any property $P$, the combining procedure outputs an estimate for
the \distinctP\ problem that is unbiased. Moreover, its variance should be bounded by that of the individual sensors' sampling procedures.\footnote{More precisely,
we are interested in showing that
the variance of the returned estimate is at most that of the (hypothetical) estimator obtained by running each individual sensor's sampling algorithm on the concatenated stream $A_1 \circ \dots \circ A_m$.
We refer to the latter estimator as ``hypothetical'' because it is typically infeasible to materialize the concatenated stream in distributed environments.} 

For reasons that will become clear later, we refer to our proposed combining procedure as the \emph{Theta-Sketch Framework}, and we refer to the mild technical condition
that each sampling procedure must satisfy to guarantee unbiasedness as \emph{$1$-Goodness}. If the sampling procedures
satisfy an additional property that we refer to as \emph{monotonicity}, then the variance of the estimate output by the combining procedure is guaranteed to satisfy the desired variance bound. 
 The Theta-Sketch Framework, and our analysis of it, unifies and generalizes a variety of results in the literature (see Section \ref{sec:contributions} for details). 

\medskip
\noindent \textbf{The Importance of Generality.} As we will see, there is a huge array of sampling procedures that the sensors could use. Each procedure comes with a unique tradeoff between 
accuracy, space requirements, update speed, and simplicity. Moreover, some of these procedures come with additional desirable properties, while others do not. We would like to support as many sampling procedures as possible, because the best one to use in any given given setting
will depend on the relative importance of each resource in that setting. 

\eat{In addition, there are realistic scenarios where different sensors may use different sampling procedures, or may use the same sampling procedure but with different settings of parameters. 
For example, if different sensors have different amounts of available memory (say, because some sensors are newer than others), then the sensors with less memory will have to downsample 
more aggressively than the others. 
The guarantees of the Theta-Sketch Framework continue to hold when different sensors use different sampling procedures, as long as all of the sampling procedures used satisfy 1-Goodness. }

\medskip
\noindent \textbf{Handling Set Expressions.}
\label{sec:setexpressions}
The scenario described above can be modeled as follows. Each sensor observes a stream of identifiers $A_j$ from a data universe of size $n$, and the goal
is to estimate the number of distinct identifiers that satisfy property $P$ in the combined stream $U=\cup_{j} A_j$. 
In full generality, we may wish to handle more complicated set expressions applied to the constituent streams, other than set-union. For example, we may have $m$ streams of identifiers $A_1, \dots, A_m$, and wish to estimate 
the number of distinct identifiers satisfying property $P$ that appear in \emph{all streams}. 
The Theta-Sketch Framework can be naturally extended to provide estimates for such queries. 
 Our analysis applies to any sequence of set operations on the $A_j$'s, but we restrict our attention to set-union and set-intersection throughout the paper for simplicity. 

\section{Preliminaries, Background, and Contributions}
\subsection{Notation and Assumptions}
\noindent \textbf{Streams and Set Operations.}
Throughout, $A$ denotes a stream of identifiers from a data universe $[n]:=\{1, \dots, n\}$. We view any \emph{property} $P$ on identifiers
as a subset of $[n]$, and let $n_{P, A} := \distinct_P(A)$ denote the number of distinct identifiers that appear in $A$ and satisfy $P$.
For brevity, we let $n_A$ denote $\distinct(A)$. 
 When working in a multi-stream setting, $A_1, \dots, A_m$ denote $m$ streams of identifiers from $[n]$, $U := \cup_{j=1}^m A_j$ will denote the concatenation of the $m$ input streams,
while $I := \cap_{j=1}^m A_j$ denotes the set of identifiers that appear at least once in all $m$ streams. Because we are interested only in \emph{distinct} counts, 
it does not matter for definitional purposes whether we view $U$ and $I$ as sets, or as multisets. 
For any property $P \colon [n] \rightarrow \{0, 1\}$, $n_{P, U} := \distinct_P(U)$ and $n_{P, I} := \distinct_P(I)$, while $n_U := \distinct(U)$ and $n_I := \distinct(I)$.

\medskip
\noindent \textbf{Hash Functions.} \label{sec:hashprelims}
For simplicity and clarity, and following prior work (e.g. \cite{beyer2009distinct, cohen2009leveraging}), we assume throughout that the sketching and sampling algorithms make use of a perfectly random hash function $h$
mapping the data universe $[n]$ to the open interval $(0, 1)$. That is, for each $x \in [n]$, $h(x)$ is a uniform random number in $(0, 1)$. 
Given a subset of hash values $S$ computed from a stream $A$, and a property $P  \subseteq [n]$,
$P(S)$ denotes the subset of hash values in $S$ whose corresponding identifiers in $[n]$ satisfy $P$.
Finally, given a stream $A$, the notation $X^{n_A}$ refers to the set of hash values obtained by mapping 
a hash function $h$ over the $n_A$ distinct identifiers in $A$.

\eat{\medskip
\noindent \textbf{Additional Notation.} Throughout, we will use the following strategy for simplifying the
presentation of our results. Let $A_{\mathrm{raw}}$ be an arbitrary stream that is allowed to contain multiple
occurrences of the various identifiers. Let $A$ be the stream derived from $A_{\mathrm{raw}}$ by deleting all
occurrences of each label except for the one that would be encountered first during 
a linear scan through $A_{\mathrm{raw}}$. Then every sketching algorithm that we consider
produces exactly the same sketch and exactly the same estimate whether it is applied to
$A_{\mathrm{raw}}$ or to $A$. 
Therefore, without loss of generality, all of our proofs concerning the contents of sketches and the values
of estimates will be written in terms of the ``uniquified'' stream $A$. 
Finally, the notation $X^{n_A}$ will always refer to the sequence of hash values obtained by mapping 
a hash function $h$ over a length-$n_A$ uniquified stream $A$.}

\subsection{Prior Art: Sketching Procedures for \distinct\ Queries}
\label{sec:priorwork}
There is a sizeable literature on streaming algorithms for estimating the number of distinct elements in a single data stream. Some, but not all,
of these algorithms can be modified to solve the \distinctP\ problem for general properties $P$. 
Depending on which functionality is required, systems based on HyperLogLog Sketches, 
K'th Minimum Value (KMV) Sketches, and Adaptive Sampling represent the
state of the art for practical systems \cite{heule2013hll}.\footnote{Algorithms with better asymptotic bit-complexity are known \cite{kane2010optimal}, but they do not match the practical
performance of the algorithms discussed here. See Section \ref{sec:related}.}
For clarity of exposition, we defer a thorough overview of these algorithms to Section \ref{app:priorwork}. Here, we briefly review the main concepts
and relevant properties of each. 

\medskip
\noindent \textbf{HLL: HyperLogLog Sketches}. HLL is a sketching algorithm for the vanilla \distinct\ problem. Its accuracy per bit 
is superior to the KMV and Adaptive Sampling algorithms described below.
However, unlike KMV and Adaptive Sampling, it is not known how to extend the HLL sketch to estimate $n_{P, A}$ for general properties $P$ (unless, of course, $P$ is known prior to stream
processing). 

\medskip
\noindent \textbf{KMV: K'th Minimum Value Sketches.} 
The KMV sketching procedure for estimating \distinctA\ works as follows. While processing an input stream $A$, KMV keeps track of the set $S$ of the $k$ smallest unique hashed values of stream elements.
The update time of a heap-based implementation of KMV is $O(\log k)$.
The KMV estimator for \distinctA\ is: 
$\romKMV_A = \; k / m_{k+1}$, where 
$m_{k+1}$ denotes the $k\!+\!1^{\text{st}}$ smallest unique 
hash value.\footnote{Some works use the estimate $k/m_k$, 
{\em e.g.} 
\cite{bar2002counting}. 
We use $k/m_{k+1}$ 
because it is unbiased, and
for consistency with the work of Cohen and Kaplan \cite{cohen2009leveraging} described below.}
It has been proved by \cite{beyer2009distinct}, \cite{giroire2009order}, and others, that $E(\romKMV_A) = n_A$, and
$\sigma^2(\romKMV_A) = \frac{n_A^2- k \; n_A}{k-1} < \frac{n_A^2}{k-1}.$
Duffield et al. \cite{DuffieldLT07} proposed to change 
the heap-based implementation of the KMV sketching algorithm
to an implementation based on quickselect \cite{quickselect61}. 
This reduces the sketch update cost from $O(\log k)$ to amortized $O(1)$. However, this $O(1)$ hides a larger constant than
competing methods.
At the cost of storing the sampled identifiers, and not just their hash values, the KMV sketching procedure can be extended to estimate $n_{P, A}$ for any property $P \subseteq [n]$ (Section \ref{app:priorwork} has details).

\medskip
\noindent \textbf{Adaptive Sampling.} 
Adaptive Sampling maintains a sampling level $i \ge 0$, and the set $S$
of all hash values less than $2^{-i}$; whenever $|S|$ exceeds a pre-specified
size limit, $i$ is incremented and $S$ is scanned discarding any hash value
that is now too big. Because a simple scan is cheaper than running
quickselect, an implementation of this scheme is typically faster than
KMV. The estimator of $n_A$ is $\romAdapt_A = \; |S| /
2^{-i}$. It has been proved by \cite{flajolet1990adaptive} that this
estimator is unbiased, and that $\sigma^2(\romAdapt_A) \approx 1.44
(n_A^2/(k-1))$, where the approximation sign hides oscillations caused
by the periodic culling of $S$.  
Like KMV, Adaptive Sampling can be extended to estimate $n_{P, A}$ for any property $P$. 
Although the stream processing speed of Adaptive Sampling is excellent,
the fact that its accuracy oscillates as $n_A$ increases is a shortcoming.

\medskip
\noindent \textbf{HLL for set operations on streams.} HLL can be directly adapted to handle set-union (see Section \ref{app:priorwork} for details). For set-intersection, the relevant adaptation uses the inclusion/exclusion principle. 
However, the variance of this estimate is approximately 
a factor of $n_U/n_I$ worse than the variance achieved by
the $\romIKMV$ algorithm described below. When $n_I \ll n_U$, this penalty
factor overwhelms HLL's fundamentally good accuracy per bit.

\medskip
\noindent \textbf{KMV for set operations on streams.} 
\edit{Given streams $A_1, \dots, A_m$, let $S_j$ denote the KMV sketch computed from stream $A_j$.
A trivial way to use these sketches to estimate the number of distinct items $n_U$ in the union stream $U$
is to let $M'_U$ denote the $(k+1)^{\text{st}}$ smallest value in the union of the sketches, and let $S'_U = \{x \in \cup_j S_j\colon x < M'_U\}$.
Then $S'_U$ is identical to the sketch that would have been obtained by running KMV directly on the concatenated stream $A_1 \circ \dots, A_m$,
and hence $\romKMV_{P,U} := k/M'_U$ is an unbiased estimator for $n_U$, by the same analysis as in the single-stream setting. 
We refer to this procedure as the ``non-growing union rule.''}

\edit{Intuitively, the non-growing union rule does not use all of the information available to it. The sets $S_j$ contain up to $k \cdot M$
distinct samples in total, but $S'_U$ ignores all but the $k$ smallest samples. 
With this in mind, Cohen and Kaplan \cite{cohen2009leveraging} proposed the following 
adaptation of KMV to handle unions of multiple streams. We denote their algorithm by $\romIKMV$, and also refer to it as the
``growing union rule''.}

For each KMV  sketch $S_j$ computed from stream $A_j$, let $M_j$ denote that sketch's value of $m_{k+1}$. 
Define $M_U = \min_{j=1}^m M_j$, and $S_U = \{x \in \cup_j S_j\colon x < M_U\}$.
Then $n_U$ is estimated by $\romIKMV_U := |S_U|/M_U$, 
and  $n_{P, U}$ is estimated by $\romIKMV_{P, U} := |P(S_U)| /M_U$. 

\edit{At first glance, it may seem obvious that the growing union rule yields an estimator that is 
``at least as good'' as the non-growing union, since the growing union rule makes use of at least as many samples
as the non-growing rule. However, it is by no means trivial to prove that $\romIKMV_{P, U}$ is unbiased,
nor that its variance is dominated by that of the non-growing union rule.}
\kedit{Nonetheless, \cite{cohen2009leveraging} managed to prove this: they showed that $\romIKMV_{P, U}$ is unbiased and has
variance that is dominated by the variance of $\romKMV_{P,U}$:}
\begin{align}
\sigma^2(\romIKMV_{P, U}) \le & \sigma^2(\romKMV_{P,U}). \label{growing-variance-dominates}
\end{align}
As observed in \cite{cohen2009leveraging}, multiKMV can be adapted in a similar manner to handle set-intersections 
\edit{(see Section \ref{sec:intersections} for details).}


\medskip
\noindent \textbf{Adaptive Sampling for set operations on streams.} 
Adaptive Sampling can handle set unions and intersections with a \edit{similar ``growing union rule''} in which ``$M_U$'' $\;:= \min_{j=1}^m (2^{-i})_j$. Here, $(2^{-i})_j$ denotes the threshold for discarding hash values that was computed by the $j$th Adaptive Sampling sketch.
We refer to this algorithm as $\romIAdapt$. 
\cite{gibbons2001estimating} proved epsilon-delta bounds on the error of $\romIAdapt_{P, U}$, but did not derive expressions for mean or variance. 
However, $\romIAdapt$ and $\romIKMV$ are both special cases
of our Theta-Sketch Framework, and in Section~\ref{sec:framework}
we will prove (apparently for the first time) that $\romIAdapt_{P, U}$ is unbiased, and satisfies strong variance bounds.
\edit{These results have the following two advantages over the epsilon-delta bounds of \cite{gibbons2001estimating}. 
First, proving unbiasedness is crucial for obtaining estimators for distinct counts over subpopulations: these estimators are analyzed as a sum of a huge number of per-item estimates (see Theorem \ref{thm:unbiased} for details), and biases add up. Second, variance bounds enable derivation of confidence intervals that an epsilon-delta guarantee cannot provide, unless the guarantee holds for many values of delta simultaneously.
}

\subsection{Overview of the Theta-Sketch Framework}
In this overview, we describe the Theta-Sketch Framework in the multi-stream setting where the goal is to output $n_{P, U}$, where $U=\cup_{j=1}^m A_j$ (we define the framework formally in Section \ref{app:framework}).
That is, the goal is to identify a very large class of sampling algorithms that can run on each constituent stream $A_j$,
as well as a ``universal'' method for combining the samples from each $A_j$ to obtain a good estimator for $n_{P, U}$.
We clarify that the Theta-Sketch Framework, and our analysis of it, yields unbiased estimators that are interesting even in the single-stream case, where $m=1$.

 We begin by noting the striking similarities between the $\romIKMV$ and
$\romIAdapt$ algorithms outlined in Section~\ref{sec:priorwork}. 
In both cases, a sketch can be viewed as pair $(\theta, S)$ where $\theta$ is a certain
threshold that depends on the stream, and $S$ is a set of hash values which are all strictly less than $\theta$.  
In this view, both schemes use the same estimator $|S|/\theta$, and also the same growing
union rule for combining samples from multiple streams.  The only difference lies in their respective rules for
mapping streams to thresholds $\theta$. 
The Theta-Sketch Framework
formalizes this pattern of similarities and differences.

\medskip
\noindent \textbf{The assumed form of the single-stream sampling algorithms.}
The Theta-Sketch Framework demands that each constituent stream $A_j$ be processed by 
a sampling algorithm $\samp_j$ of the following form. While processing $A_j$, $\samp_j$
evaluates a ``threshold choosing function'' (TCF) $T^{(j)}(A_j)$. The final state of $\samp_j$ must be of the form
$(\theta_j:=T^{(j)}(A_j), S)$, where $S$ is the set of all hash values strictly less than $\theta_j$ that were observed while processing $A_j$. 
If we want to estimate $n_{P, U}$ for non-trivial properties $P$, then $\samp_j$
must also store the corresponding identifier that hashed to each value in $S$.
Note that the framework itself does not specify the threshold-choosing functions $T^{(j)}$. Rather, any specification of the TCFs
$T^{(j)}$ defines a particular instantiation of the framework.

{\small
\begin{algorithm}[t]
{\small
\caption{\small{Theta Sketch Framework for estimating $n_{P, U}$. The framework is parameterized by 
choice of TCF's $T^{(j)}$($k$,$A_j$,$h$), one for each input stream.}}\label{code:framework}
\begin{algorithmic}[1]{\footnotesize
\vspace{0.5em}
\STATE \textbf{Definition:} Function samp$_j$[$T^{(j)}$]($k$, $A_j$, $h$)
\bindent
 \STATE  $\theta_j \leftarrow \mathrm{T}^{(j)}(k,A_j,h)$
\STATE $S_j \leftarrow \{ (x \in h(A_j)) < \theta_j \}$. \label{code:framework-line-3}
\RETURN $(\theta_j,S_j)$.
\eindent
\vspace{0.5em}
\STATE \textbf{Definition:} Function ThetaUnion(Theta Sketches $\{ (\theta_j,S_j) \}$)
\bindent
 \STATE  $\theta_U \leftarrow \min \{\theta_j\}$. 
 \STATE $S_U \leftarrow \{ (x \in (\cup S_j)) < \theta_U \}$. \label{code:framework-line-7}
 \RETURN $(\theta_U,S_U)$.
 \eindent
\vspace{0.5em}
\STATE \textbf{Definition:} Function EstimateOnSubPopulation(Theta Sketch $(\theta,S)$ produced from stream $A$, Property $P$ mapping identifiers to $\{0, 1\}$)
\bindent \RETURN  $\hat{n}_{A,P} := \frac{|P(S)|}{\theta}$.
\eindent
}\end{algorithmic}
}
\end{algorithm}
}

\medskip
\noindent \textbf{Remark.} It might appear from Algorithm~\ref{code:framework} that for any TCF $T^{(j)}$, the function $\samp_j[T^{(j)}]$ makes two passes over the input stream: one to compute $\theta_j$, and another to compute $S_j$.  However, in all of the instantiations we consider, both operations can be performed in a single pass.

\medskip
\noindent \textbf{The universal combining rule.} Given the states $(\theta_j:=T^{(j)}(A_j), S_j)$ of each of the $m$ sampling algorithms when run on the streams $A_1, \dots, A_m$,
define $\theta_U := \min_{j=1}^{m} \theta_j$, and $S_U := \{x \in \cup_j S_j \colon x < \theta_U\}$ (see the function $\thetaunion$ in Algorithm \ref{code:framework}). 
Then $n_U$ is estimated by $\romTS_U := |S_U|/\theta_U$, and $n_{P, U}$ as 
$\romTS_{P, U} := 
|P(S_U)| / \theta_U$
(see the function $\estonsub$ in Algorithm \ref{code:framework}).

\medskip
\noindent \textbf{The analysis.} Our analysis shows that, so long as each threshold-choosing function $T^{(j)}$ satisfies a mild technical condition that we 
call \emph{1-Goodness}, then $\romTS_{P, U}$ is unbiased. We also show that if each $T^{(j)}$ satisfies a certain additional condition that we call \emph{monotonicity}, then 
$\romTS_{P, U}$ satisfies strong variance bounds (analogous to the bound of Equation \eqref{growing-variance-dominates} for $\romKMV$).  
Our analysis is arguably surprising, because $1$-Goodness does not imply certain properties that have
traditionally been considered important, such as permutation invariance, or $S$ being a uniform
random sample of the hashed unique items of the input stream.

\medskip
\noindent \textbf{Applicability.}  To demonstrate the generality of our analysis, we identify several valid instantiations of the Theta-Sketch Framework. 
First, we show that the TCF's used in KMV and Adaptive Sampling both satisfy $1$-Goodness and monotonicity, implying that $\romIKMV$ and $\romIAdapt$ are both unbiased and 
satisfy the aforementioned variance bounds. For $\romIKMV$, this is a reproof
of Cohen and Kaplan's results \cite{cohen2009leveraging}, but for $\romIAdapt$ the results are new.
Second, we identify a variant of KMV that we call $\pIKMV$,  which is useful in multi-stream settings where the lengths of constituent streams are highly skewed.
We show that $\pIKMV$ satisfies both $1$-Goodness and monotonicity. 
Third, we introduce a new sampling procedure that we call the \emph{Alpha Algorithm}. Unlike earlier algorithms, the Alpha Algorithm's
 final state actually depends on the stream order, yet we  
 show that it satisfies $1$-Goodness, and hence
is unbiased in both the single- and multi-stream settings.
We also establish variance bounds on the Alpha Algorithm in the single-stream setting.
We show experimentally that the Alpha Algorithm, in both the single- and multi-stream settings, 
achieves a novel tradeoff between accuracy, space usage, update speed, and applicability. 

Unlike KMV and Adaptive Sampling, 
 the Alpha Algorithm does not satisfy monotonicity in general. 
In fact, we have identified contrived examples
in the multi-stream setting on which the aforementioned variance bounds are (weakly) violated.
The Alpha Algorithm does, however, satisfy monotonicity under the promise that the $A_1, \dots, A_m$ are pairwise disjoint, implying variance bounds in this case. 
Our experiments suggest that, in practice, the \edit{normalized} variance in the multi-stream setting is not much larger than in the pairwise disjoint case.
\edit{\subparagraph*{Deployment of Algorithms.}  Within Yahoo, the pKMV and Alpha algorithms are used widely. In particular,
stream cardinalities in Yahoo empirically satisfy a power law, with some very large streams and many short ones, and pKMV is an attractive option for such settings. 
We have released an optimized open-source implementation of our algorithms at \texttt{http://datasketches.github.io/}.}

\subsection{Formal Definition of Theta-Sketch Framework}
\label{app:framework}
The Theta-Sketch Framework is defined as follows. This definition is specific to the multi-stream setting where the goal is to output $n_{P, U}$, where $U=\cup_{j=1}^m A_j$ is the union of constituent streams $A_1, \dots, A_m$.
\begin{definition} \label{def:tsf}
The Theta-Sketch Framework consists of the following components:
 \begin{itemize}
  \item The data type $(\theta, S)$, where $0 < \theta \le 1$ is a 
   threshold, and $S$ is the set of all unique hashed stream items
   $0 \le x < 1$ that are less than $\theta$. We will generically use the term ``theta-sketch'' to refer to an instance
   of this data type.
  \item The universal ``combining function'' $\thetaunion()$, defined in Algorithm~\ref{code:framework}, that takes as input a 
  collection of theta-sketches (purportedly obtained by running $\mathrm{\samp}[T]$() on constituent streams $A_1, \dots, A_m$),
  and returns a single theta-sketch (purportedly of the union stream $U = \cup_{i=1}^m A_i$).
  \item The function $\estonsub()$, defined in Algorithm~\ref{code:framework},
  that takes as input a theta-sketch $(\theta,S)$ (purportedly obtained from some stream $A$) and a property $P\subseteq [n]$
  and returns an estimate of $\romTS_{P, A}$.
  \end{itemize}
  Any instantiation of the Theta-Sketch Framework must specify a ``threshold choosing function'' (TCF), denoted $T(k,A, h)$,
   that maps a target sketch size, a stream, and a hash function $h$ to a threshold $\theta$. 
   Any TCF $T$ implies a ``base'' sampling procedure $\mathrm{\samp}[T]$() that maps a target size, a stream $A$, and a hash function to a theta-sketch
  using the pseudocode shown in Algorithm~\ref{code:framework}. One can obtain an estimate $\romTS_{P, A}$ for $n_{P, A}$ by feeding the resulting theta-sketch
  into $\estonsub$().
  
  Given constituent streams $A_1, \dots, A_m$, the instantiation obtains an estimate $\romTS_{P, U}$ of $n_{P, U}$ by running $\mathrm{\samp}[T]$() 
  on each constituent stream $A_j$, feeding the resulting theta-sketches to $\thetaunion$() to obtain a ``combined'' theta-sketch for $U=\cup_{i=1}^m A_i$, and then running
  $\estonsub$()  on this combined sketch. 
\end{definition}

\medskip \noindent \textbf{Remark.}
Definition \ref{def:tsf} assumes for simplicity that the same TCF $T$ is used in the base sampling algorithms run 
on each of the constituent streams. However, all of our results that depend only on $1$-Goodness ({\em e.g.} unbiasedness
of estimates and non-correlation of ``per-item estimates'') hold even if different $1$-Good
TCF's are used on each stream, and even if different values of $k$ are employed.

\subsection{Summary of Contributions}
\label{sec:contributions}
In summary, our contributions are: (1) Formulating the Theta-Sketch Framework. (2) Identifying a mild technical condition ($1$-Goodness) on TCF's 
ensuring that the framework's estimators 
are unbiased. If each TCF 
also satisfies a monotonicity condition, the framework's estimators come with strong variance bounds analogous to Equation \eqref{growing-variance-dominates}. 
(3) Proving $\romIKMV$, $\romIAdapt$, and $\pIKMV$ all satisfy $1$-Goodness and monotonicity, implying unbiasedness and variance bounds for each.
(4) Introducing the Alpha Algorithm, proving that it is unbiased, 
and establishing quantitative bounds on its variance in the single-stream setting.
(5) Experimental results showing that the Alpha Algorithm instantiation achieves a novel tradeoff between accuracy, space usage, update speed, and applicability.


\section{Analysis of the Theta-Sketch Framework}
\label{sec:framework}
\noindent \textbf{Section Outline.} 
Section~\ref{sec:exampleinstantiations} shows that KMV and Adaptive Sampling are both instantiations of the Theta-Sketch Framework.
Section~\ref{sec:sufficient-condition} defines $1$-Goodness.
Sections~\ref{sec:kmv-sat} and~\ref{sec:adapt-sat} prove that the TCF's 
that instantiate behavior identical to $\romKMV$ and $\romAdapt$ both satisfy $1$-Goodness.
Section~\ref{sec:union-sat} proves that if a framework instantiation's
TCF satisfies $1$-Goodness, then so does the TCF that is implicitly applied to the union stream via the composition of the 
instantiation's base algorithm and the function $\thetaunion$(). 
Section~\ref{sec:subpop-unbiased} proves that the estimator $\romTS_{P, A}$ for $n_{P, A}$ returned by $\estonsub$()
is unbiased when applied to any theta-sketch produced by a TCF
satisfying $1$-Goodness. 
Section~\ref{sec:variance} defines monotonicity and shows that $1$-Goodness and monotonicity together imply variance bounds on $\romTS_{P, U}$. 
Section \ref{sec:intersections} explains how to tweak the Theta-Sketch Framework
to handle set intersections and other set operations on streams.
Finally, Section \ref{sec:pkmv} describes the $\pIKMV$ variant of KMV. 

\subsection{Example Instantiations}
\label{sec:exampleinstantiations} 
%
Define $m_{k+1}$ to be the $k\!+\!1^{\text{st}}$ smallest unique hash value in 
$h(A)$ (the hashed version of the input stream).
The following is an easy observation. 
\begin{observation}
\label{obs:kmv}
When the Theta-Sketch Framework is instantiated with the TCF $T(k,A,h) = m_{k+1}$, 
the resulting instantiation is equivalent to the $\romIKMV$
algorithm outlined in  
Section~\ref{sec:priorwork}.
\end{observation}
Let $\beta$ be any real value in $(0,1)$. For any $z$, define $\beta^{i(z)}$ to be the 
largest value of $\beta^i$ (with $i$ a
non-negative integer) that is less than $z$. 
\begin{observation} 
\label{obs:adapt}
When the Theta-Sketch Framework is instantiated with the TCF $T(k,A,h) = \beta^{i(m_{k+1})}$
the resulting instantiation is equivalent to $\romIAdapt$, which 
combines Adaptive Sampling with a growing union rule (cf. Section~\ref{sec:priorwork}).\footnote{ 
Section~\ref{sec:priorwork} assumed that the parameter $\beta$ was set to the most common value: $1/2$.}
\end{observation}

\subsection{Definition of $1$-Goodness}\label{sec:sufficient-condition}

The following circularity is a main source of technical difficulty in analyzing theta
sketches:
for any given identifier $\ell$ in a stream $A$, whether its hashed value $x_\ell = h(\ell)$
will end up in a sketch's sample set $S$ depends on a comparison of $x_\ell$ versus a threshold $T(X^{n_A})$ that
depends on 
$x_\ell$ itself.
Adapting a technique from \cite{cohen2009leveraging}, we partially break this circularity by analyzing 
the following infinite family of projections of a given threshold choosing function $T(X^{n_A})$. 

\begin{definition}[Definition of Fix-All-But-One Projection]\label{def:fabo-projection}
Let $T$ be a threshold choosing function.
Let $\ell$ be one of the $n_A$ unique identifiers in a stream $A$. Let $\xnml$ be a fixed assignment of 
hash values to all unique identifiers in $A$ {\em except} for $\ell$. 
Then the fix-all-but-one projection $T_\ell[\xnml](x_\ell) : (0,1) \rightarrow (0,1]$ of $T$
is the function that maps values of $x_\ell$ to theta-sketch thresholds via the definition
$T_\ell[\xnml](x_\ell) = T(X^{n_A}),$ where $X^{n_A}$ is the obvious combination of $\xnml$ and $x_\ell$.
\end{definition}
\cite{cohen2009leveraging} analyzed similar projections under the assumption
that the base algorithm is specifically (a weighted version of) KMV; we will instead 
impose the weaker condition that every fix-all-but-one projection satisfies $1$-Goodness, defined below.%
\footnote{We chose the name $1$-Goodness due to the reference to Fix-All-But-\emph{One} Projections.}

\begin{definition}[Definition of $1$-Goodness for Univariate Functions]\label{def:good-shape}
A function $f(x): (0,1) \rightarrow (0,1]$ satisfies $1$-Goodness iff there exists a fixed threshold $F$ such that:
\begin{align}
\text{If} \;\; x  <  F, & \;\; \mathrm{then} \;\; f(x)  = F. \label{shape-condition-a} \\
\text{If} \;\; x \ge F, & \;\; \mathrm{then} \;\; f(x) \le x. \label {shape-condition-b}
\end{align}
\end{definition}

\noindent Figure~\ref{fig:six-shape-plots}  contains six examples of hypothetical projections of TCF's. 
Four of them satisfy $1$-Goodness; the other two do not. 

\label{app:figures}
\begin{figure}
\begin{center}
\includegraphics[width=.5\linewidth]{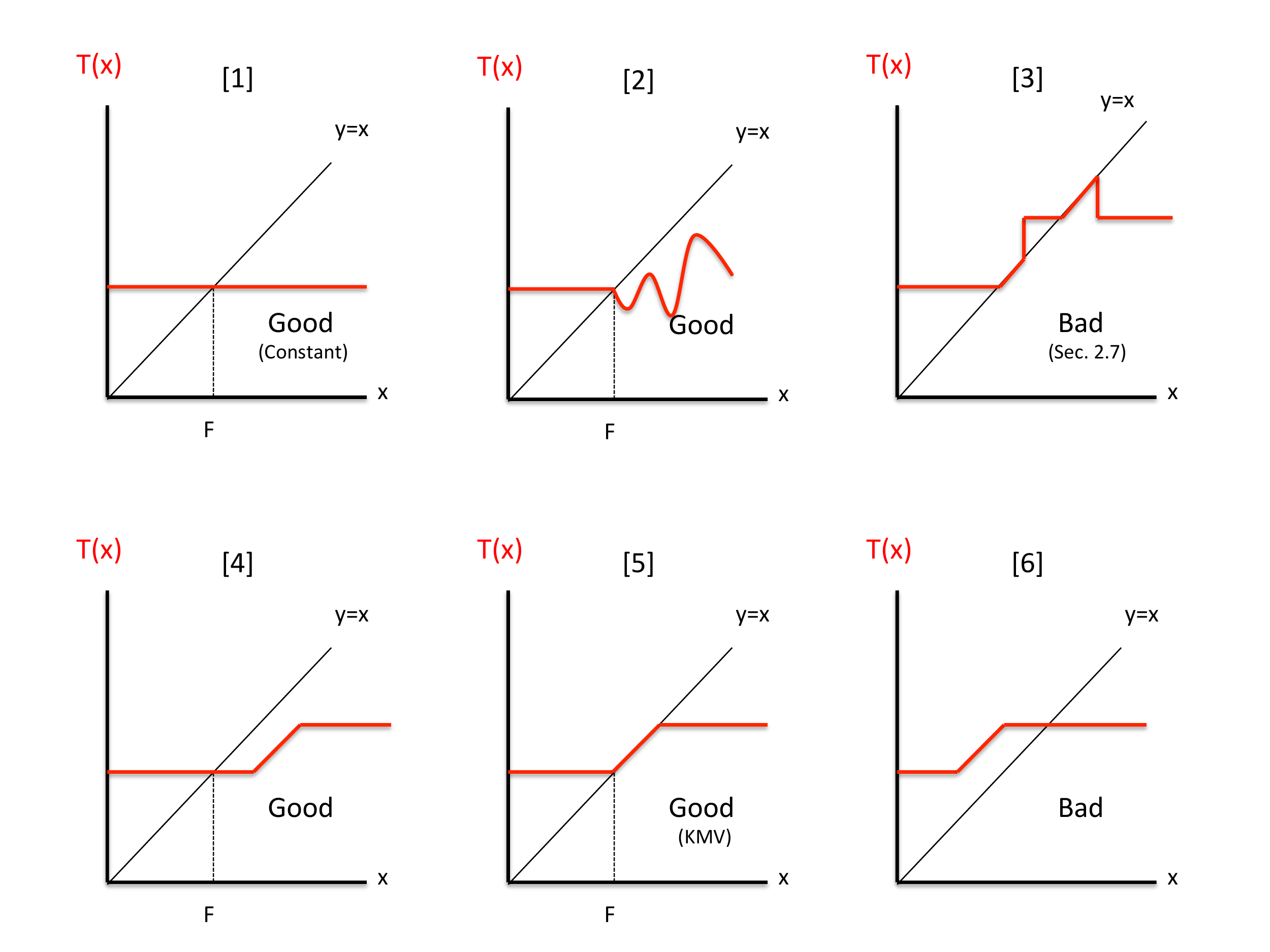}
\end{center}
\caption{Six examples of hypothetical projections of TCF's. Four of them satisfy $1$-Goodness; the other two do not.}
\label{fig:six-shape-plots}
\end{figure}


\begin{condition}[Definition of $1$-Goodness for TCF's]\label{key-condition}
A TCF $T(X^{n_A})$ satisfies $1$-Goodness iff for every stream $A$
containing $n_A$ unique identifiers, every label $\ell \in A$, and every fixed 
assignment $X^{n_A}_{-\ell}$ of hash values to the identifiers in $A \!\setminus\! \ell$,
the fix-all-but-one projection $\txn$ satisfies Definition~\ref{def:good-shape}.
\end{condition}

\subsection{TCF of $\romIKMV$ Satisfies $1$-Goodness}\label{sec:kmv-sat}
The following theorem shows that the TCF used in KMV satisfies $1$-Goodness. 
\begin{theorem}\label{thm:kmv-sat}
If $T(X^{n_A}) = m_{k+1}$, then every fix-all-but-one projection $\txn$ of $T$ satisfies $1$-Goodness.
\end{theorem}
\begin{proof}
Let $\txn$ be any specific fix-all-but-one-projection of $T(X^{n_A}) = m_{k+1}$.
We will exhibit the fixed value $\xhatfull$ that causes (\ref{shape-condition-a}) and (\ref{shape-condition-b})
to be true for this projection.
Let $a$ and $b$ respectively be the $k$'th and $(k\!+\!1)^{\text{st}}$
smallest hash values in $\xnml$. 
Then Subconditions (\ref{shape-condition-a}) and (\ref{shape-condition-b}) hold for
$\xhatfull = a$. There are three cases:
\begin{tighterdescription}
\item[Case $(x_\ell < a < b):\;\;$] 
In this case, 
$\txn = T(X^{n_A}) = m_{k+1} = a$. 
Since $x_\ell < (\xhatfull = a)$, 
(\ref{shape-condition-a}) holds because $(\txn = a) = \xhatfull$, and (\ref{shape-condition-b}) holds vacuously.
\item[Case $(a < x_\ell < b):$] 
In this case, 
$\txn = T(X^{n_A}) = m_{k+1} = x_\ell$. 
Since $ x_\ell \ge (\xhatfull = a)$, 
(\ref{shape-condition-b}) holds because  $(\txn = x_\ell) \le x_\ell$, and (\ref{shape-condition-a}) holds vacuously.
\item[Case $(a < b < x_\ell):\;\;$] 
In this case, 
$\txn = T(X^{n_A}) = m_{k+1} = b$. 
Since $ x_\ell \ge (\xhatfull = a)$, 
(\ref{shape-condition-b}) holds because  $(\txn = b) < x_\ell$, and (\ref{shape-condition-a}) holds vacuously.
\end{tighterdescription}
\end{proof}


\subsection{TCF of $\romIAdapt$ Satisfies $1$-Goodness}\label{sec:adapt-sat}
The following theorem shows that the TCF used in Adaptive Sampling satisfies $1$-Goodness.
\begin{theorem}\label{thm:adapt-sat}
If $T(X^{n_A}) = \beta^{i(m_{k+1})}$, 
then every fix-all-but-one projection $\txn$ of $T$ satisfies $1$-Goodness.
\end{theorem}
\begin{proof}
Let $\txn$ be any specific fix-all-but-one-projection of $T(X^{n_A}) = \beta^{i(m_{k+1})}$.
We will exhibit the fixed value $\xhatfull$ that causes (\ref{shape-condition-a}) and (\ref{shape-condition-b})
to be true for this projection.
Let $a$ and $b$ respectively be the $k$'th and $(k\!+\!1)^{\text{st}}$
smallest hash values in $\xnml$. 
Then Subconditions (\ref{shape-condition-a}) and (\ref{shape-condition-b}) hold for
$\xhatfull = \beta^{i(a)}$. There are four cases:

\begin{tighterdescription}
\item[Case $(x_\ell < \beta^{i(a)} < a < b):\;\;$]
$m_{k+1} = a$, so $\txn = \beta^{i(a)}$.
Since $x_\ell < (\xhatfull = \beta^{i(a)})$, (\ref{shape-condition-a}) holds because  $(\txn = \beta^{i(a)}) = \xhatfull$,
and (\ref{shape-condition-b}) holds vacuously.

\item[Case $(\beta^{i(a)} < x_\ell < a < b):\;\;$]
$m_{k+1} = a$, so $\txn = \beta^{i(a)}$.
Since $x_\ell \ge (\xhatfull = \beta^{i(a)})$, (\ref{shape-condition-b}) holds because  $(\txn = \beta^{i(a)}) < x_\ell$, 
and (\ref{shape-condition-a}) holds vacuously.

\item[Case $(\beta^{i(a)} < a < x_\ell < b):\;\;$]
$m_{k+1} = x_\ell$, so $\txn = \beta^{i(x_\ell)}$.
Since $x_\ell \ge (\xhatfull = \beta^{i(a)})$, (\ref{shape-condition-b}) holds because $(\txn = \beta^{i(x_\ell)}) < x_\ell$,
and (\ref{shape-condition-a}) holds vacuously.

\item[Case $\beta^{i(a)} < a < b < x_\ell):\;\;$]
$m_{k+1} = b$, so $\txn = \beta^{i(b)}$.
Since $x_\ell \ge (\xhatfull = \beta^{i(a)})$, (\ref{shape-condition-b}) holds because  $(\txn = \beta^{i(b)}) < b < x_\ell$,
and (\ref{shape-condition-a}) holds vacuously.
\end{tighterdescription}
\end{proof}

\subsection{$1$-Goodness Is Preserved by the Function $\thetaunion()$}\label{sec:union-sat}
Next, we show that if a framework instantiation's TCF $T$ satisfies $1$-Goodness, 
then so does the TCF $T^U$ that is implicitly being used by the theta-sketch construction 
algorithm defined by the composition of the instantiation's base sampling algorithms and the function $\thetaunion$().
We begin by formally extending the definition of a fix-all-but-one projection to cover 
the degenerate case where the label $\ell$ isn't actually a member of the given stream $A$.

\begin{definition}\label{def:extended-projection}
Let $A$ be a stream containing $n_A$ identifiers. Let $\ell$ be a label that is {\em not} a member of $A$. 
Let the notation $\xnml$ refer to an assignment of hash value to {\em all} identifiers in $A$.
For any hash value $x_\ell$ of the non-member label $\ell$,
define the value of the ``fix-all-but-one'' projection
$T_\ell[\xnml](x_\ell)$ to be the constant $T(\xnml)$.
\end{definition}

\begin{theorem}\label{thm-union-preserves-condition} \label{thm:preserved}
If the threshold choosing functions $T^{(j)}(\xnj)$ of the base algorithms used to create sketches 
of $m$ streams $A_j$ all satisfy Condition~\ref{key-condition},
then so does the TCF:
\begin{align}
T^U(\xon) = \min_j \{T^{(j)}(\xnj)\}\label{eqn-define-tcfu}
\end{align}
that is implicitly applied to the union stream via the composition of those base
algorithms and the procedure $\thetaunion$().
\end{theorem}
\begin{proof}
Let $\tuxn$ be any specific fix-all-but-one projection of the threshold choosing function 
$T^U(\xnu)$ defined by Equation (\ref{eqn-define-tcfu}). We will exhibit the fixed value $F^U[\xnuml]$ that
causes (\ref{shape-condition-a}) and (\ref{shape-condition-b}) to be true for $\tuxn$.

The projection $\tuxn$ is specified by a label $\ell \in (A_U = \cup_j A_j)$, and a set $\xnuml$ 
of fixed hash values for the identifiers in 
\mbox{$A_U \!\! \setminus \! \ell$}.
For each $j$, those fixed hash values $\xnuml$ induce a set $\xnjml$ of 
fixed hash values for the identifiers in $A_j \! \setminus \! \ell$. 
The combination of $\ell$ and $\xnjml$ then specifies
a projection $\tjxn$ of $T^{(j)}(X^j)$. 
Now, if $\ell \in A_j$, this is a fix-all-but-one projection according to the original Definition \ref{def:fabo-projection},
and according to the current theorem's pre-condition, this projection must satisfy $1$-Goodness for univariate functions.
On the other hand, if $\ell \not \in A_j$, this is a fix-all-but-one projection according to the extended Definition 
\ref{def:extended-projection}, and is therefore a constant function, and therefore 
satisfies $1$-Goodness.
Because the projection $\tjxn$ 
satisfies $1$-Goodness
either way, there must exist a fixed value $F^j[\xnjml]$
such that Subconditions (\ref{shape-condition-a}) and (\ref{shape-condition-b}) are true for $\tjxn$. 

\noindent We now show that the value $\afuxn := \min_j (\afjxn)$
causes Subconditions (\ref{shape-condition-a}) and (\ref{shape-condition-b}) to be true for the projection
$\tuxn$, thus proving that this projection
satisfies $1$-Goodness.

\medskip 
\noindent {\bf To show:} 
$x_\ell < \afuxn \; \mathrm{implies} \; \atuxn = \afuxn $. 
The condition $x_\ell < \afuxn$ implies
that  for all $j$, $x_\ell < \afjxn$. Then, for all $j$, $\atjxn = \afjxn$ 
by Subcondition (\ref{shape-condition-a}) for the various $\atjxn$. Therefore,
$\afuxn = \min_j (\afjxn) = \min_j (\atjxn) = \atuxn$,
where the last step is by Eqn~(\ref{eqn-define-tcfu}).
This establishes Subcondition (\ref{shape-condition-a}) for the projection $\atuxn$. 

\medskip
\noindent {\bf To show:} 
$x_\ell \ge \afuxn \; \mathrm{implies} \; x_\ell \ge \atuxn $. Because $x_\ell \ge \afuxn = \min_j (\afjxn)$,
there exists a $j \;\text{such that} \; x_\ell \ge \afjxn$. 
By Subcondition (\ref{shape-condition-b}) for this $\atjxn$, we have $x_\ell \ge \atjxn$. By
Eqn~(\ref{eqn-define-tcfu}), we then have $ x_\ell \ge \atuxn $, thus establishing
Subcondition (\ref{shape-condition-b}) for $\atuxn$. 

Finally, because the above argument applies to every
projection $\tuxn$ of $T^U(\xon)$, we have proved the desired 
result that $T^U(\xon)$ satisfies condition~\ref{key-condition}.
\end{proof}


\subsection{Unbiasedness of $\estonsub$()}\label{sec:subpop-unbiased}
We now show that $1$-Goodness of a TCF implies that the corresponding instantiation 
of the Theta-Sketch Framework provides unbiased estimates of the number of unique identifiers
on a stream or on the union of multiple streams.
\begin{theorem}\label{thm-key-property-implies-unbiased-for-single-streams} \label{thm:unbiased}
Let $A$ be a stream containing $n_A$ unique identifiers, and let $P$ be a property evaluating to $1$ on an arbitrary subset of the identifiers.
Let $h$ denote a random hash function.
Let $T$ be a threshold choosing function that satisfies Condition~\ref{key-condition}. 
Let $(\theta,S_A)$ denote a sketch of $A$ created by $\mathrm{\samp}[T](k,A,h)$,
and as usual let $P(S_A)$ denote the subset of hash values in $S_A$ whose corresponding identifiers satisfy $P$.
Then 
$\ept_h\left(\romTS_{P, A}\right) := \ept_h\left(
\frac{|P(S_A)|}{\theta}
\right) = n_{P, A}.$
\end{theorem}

\noindent Theorems \ref{thm-union-preserves-condition} and \ref{thm-key-property-implies-unbiased-for-single-streams}  together
imply that, in the multi-stream setting, the estimate $\romTS_{P, U}$ for $n_{P, U}$ output by the Theta-Sketch Framework is unbiased,
assuming the base sampling schemes $\samp_j$() each use a TCF $T^{(j)}$ satisfying $1$-Goodness.

\begin{proof}
Let $A$ be a stream, and let $T$ be a Threshold Choosing Function that satisfies $1$-Goodness. Fix any $\ell \in A$. For any assignment $\xn$ of hash values to identifiers in $A$,
define the ``per-identifier estimate'' $V_\ell$ as follows:

\begin{equation}\label{define-v-ell}
V_\ell(\xn) = \frac{S_\ell(\xn)}{T(\xn)}\quad \mathrm{where} \quad 
S_\ell(\xn) = \left\{
\begin{array}{l}
1 \;\mathrm{if}\; x_\ell < T(\xn) \\
0 \;\mathrm{otherwise}.
\end{array}\right.
\end{equation}

\noindent Because $T$ satisfies $1$-Goodness, there exists a fixed threshold $F(\xnml)$ for which it is
a straightforward exercise to verify that:

\begin{equation}\label{v-in-terms-of-f}
V_\ell(\xn) = \left\{
\begin{array}{l}
1 / F(\xnml)
\;\mathrm{if}\; x_\ell < F(\xnml) \\
0 \;\mathrm{otherwise}.
\end{array}\right.
\end{equation}

\noindent Now, conditioning on $\xnml$ and taking the expectation with respect to $x_\ell$:

\begin{equation} \label{expectedeq}
E(V_\ell | \xnml) = \int_0^1 V_\ell[\xn](x_\ell) dx_\ell = F(\xnml) \cdot \frac{1}{F(\xnml)} = 1.
\end{equation}

\noindent Since Equation \eqref{expectedeq} establishes that $E(V_\ell) = 1$ when conditioned on each $\xnml$, we also have
$E(V_\ell) = 1$ when the expectation is taken over all $\xn$.
By linearity of expectation, we conclude that
$E(\hat{n}_{P,A}) = \sum_{\ell \in A : P(\ell) = 1} E(V_\ell) = n_{P,A}.$
\end{proof}

\noindent{\textbf{Is $1$-Goodness Necessary for Unbiasedness?}}\label{section:negative-example}
Here we give an example showing
that $1$-Goodness cannot be substantially weakened while still guaranteeing unbiasedness of the estimate 
 $\romTS_{P, U}$
returned by the Theta-Sketch Framework. 
By construction, the following threshold choosing function causes the 
estimator of the Theta-Sketch Framework to be biased upwards. 
\begin{equation}
\mathrm{T}(X^{n_A}) = 
\left\{
\begin{array}{l}
m_k     \;\mathrm{if}\; \frac{k-1}{m_k} > \frac{k}{m_{k+1}} \\
m_{k+1} \;\mathrm{otherwise}
\end{array}\right.
\end{equation}
Therefore, by the contrapositive of Theorem~\ref{thm:unbiased},
it cannot satisfy Condition \ref{key-condition}. 
It is an interesting exercise to try to 
establish this fact directly. It can be done by exhibiting a specific
target size $k$, stream $A$, and partial assignment of hash values $\xnml$ 
such that no fixed threshold $\xhatfull$ exists that would satisfy
(\ref{shape-condition-a}) and (\ref{shape-condition-b}). Here is one such 
example: $k=3$, $h(A) = \{0.1, 0.2, 0.4, 0.7, x_\ell\}$. 
\begin{center}
\includegraphics[width=0.35\linewidth]{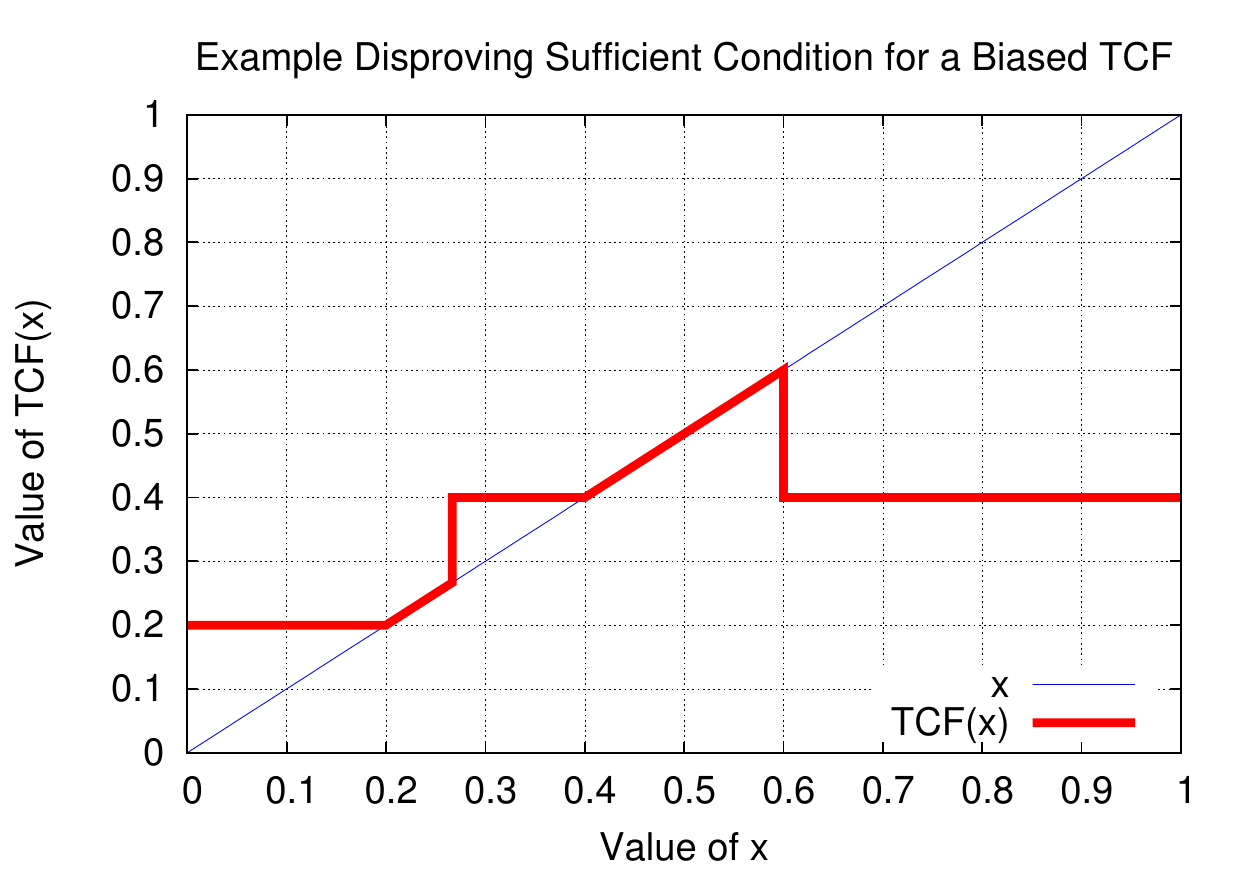}
\end{center}
The non-existence of the required fixed threshold is proved
by the above plot of $T(x_\ell)$. The only
value of $\xhatfull$ that would satisfy subcondition (\ref{shape-condition-a}) is 0.2.
However, that value does {\em not} satisfy (\ref{shape-condition-b}), 
because $T(x_\ell) > x_\ell$ for $8/30 < x_\ell < 0.4$.

\subsection{$1$-Goodness and Monotonicity Imply Variance Bound}
\label{sec:variance}

As usual, let $U=\cup_{i=1}^m A_i$ be the union of $m$ data streams. 
Our goal in this section is to identify conditions on a threshold choosing function which guarantee the following: whenever the Theta-Sketch Framework is instantiated with a TCF $T$ satisfying the conditions,
then for any property $P \subseteq [n]$, the variance $\sigma^2(\romTS_{P, U})$ of the estimator obtained from the Theta-Sketch Framework
is bounded above by 
the variance of the estimator obtained by running $\samp[T]$() on the stream $A^* := A_1 \circ A_2 \circ \dots \circ A_m$ obtained by concatenating $A_1, \dots, A_m$. 

It is easy to see that $1$-Goodness alone is not sufficient to ensure such a variance bound.
Consider, for example, a TCF $T$ that runs KMV 
on a stream $A$ unless it determines that $n_A \geq C$, for some fixed value $C$, at which point it sets $\theta$ to $1$ (thereby causing $\samp[T]$() to sample all elements from $A$). 
Note that such a base sampling algorithm is not implementable by a sublinear space streaming algorithm, but $T$ nonetheless satisfies $1$-Goodness.
It is easy to see that such a base sampling algorithm will fail to satisfy our desired comparative variance result when run on constituent streams $A_1, \dots, A_m$ satisfying $n_{A_i} < C$ for all $i$, and $n_U > C$. 
In this case, the variance of 
$\romTS_U$ will be positive, while the variance of the estimator obtained by running $\samp[T]$ directly on $A^*$ will be 0. 

Thus, for our comparative variance result to hold, we assume that $T$ satisfies both $1$-Goodness and the following additional monotonicity condition.

\begin{condition}[Monotonicity Condition] \label{mildcondition}
Let $A_0, A_1, A_2$ be any three streams, and let $A^* := A_0 \circ A_1 \circ A_2$ denote their concatenation. Fix any hash function $h$ and parameter $k$. Let $\theta=T(k, A_1, h)$, and $\theta'=T(k, A^*, h)$. Then $\theta' \leq \theta$. 
\end{condition}
\begin{theorem} \label{thm:variance}
Suppose that the Theta-Sketch Framework is instantiated with a TCF $T$ that satisfies Condition~\ref{key-condition} ($1$-Goodness), as well as Condition \ref{mildcondition} (monotonicity). Fix a property $P$, and
 let $A_1$,
\ldots $A_m$, be $m$ input streams.  Let $U = \cup A_j$ denote the union of the distinct labels in the input streams.
Let $A^* = A_1 \circ A_2 \circ \ldots \circ A_m$ denote the concatenation of the input streams.
Let $(\theta^*,S^*) = \samp[T](k,A^*,h)$, and let $\hat{n}^{A^*}_{P,A^*}$  denote the 
estimate of $n_{P,A^*}=n_{P, U}$ obtained by 
evaluating $\estonsub((\theta^*, S^*),P)$.
Let $(\theta^U,S^U) = ThetaUnion(\{(\theta_j,S_j)\})$, and 
let $\hat{n}^U_{P,U}$ 
denote the estimate of $n_{P,U} = n_{P,A^*}$ obtained by 
evaluating $\estonsub((\theta^U,S^U),P)$.
Then, with the randomness being over the choice of hash function $h$,
$\sigma^2(\hat{n}^U_{P,U}) \le \sigma^2(\hat{n}^{A^*}_{P,A^*}).$
\end{theorem}
The proof of Theorem \ref{thm:variance} is somewhat involved, and is deferred to Appendix \ref{app:variance}. 

\medskip
\noindent \textbf{On the applicability of Theorem \ref{thm:variance}.}
\label{sec:applicability}
It is easy to see that Condition \ref{mildcondition} holds for any TCF that is (1) order-insensitive and (2)
has the property that adding another distinct item to the stream cannot increase the resulting threshold $\theta$.
The TCF $T$ used in $\romIKMV$ (namely, $T(k,A,h) = m_{k+1}$), satisfies these properties, as does
the TCF used in Adaptive Sampling. Since we already showed that both of these TCF's satisfy $1$-Goodness, Theorem \ref{thm:variance} applies
to $\romIKMV$ and $\romIAdapt$.
In Section \ref{sec:pkmv}, we introduce the $\pIKMV$ algorithm, which is useful in multi-stream settings where the distribution of stream lengths is highly skewed, and we
show that Theorem \ref{thm:variance} applies to this algorithm as well. 


In Section \ref{sec:alpha},
we introduce the Alpha Algorithm and show that it satisfies $1$-Goodness. Unfortunately, the Alpha Algorithm 
does not satisfy monotonicity in general. 
The algorithm does, however, satisfy monotonicity under the promise that $A_1, \dots, A_m$ are pairwise disjoint, and Theorem \ref{thm:variance} applies in this case. 
Our experiments (Section \ref{app:webscope-experiment}) suggest that, in practice, the \edit{normalized} 
variance in the multi-stream setting is not much larger than in the pairwise disjoint case.
%
%
%
\subsection{Handling Set Intersections}\label{sec:intersections}
The Theta-Sketch Framework can be tweaked in a natural way to handle set intersection and other set operations, just as was the case
for $\romIKMV$ (cf. Section \ref{sec:IKMV}).
Specifically, define $\theta_U = \min_{j=1}^m \theta_j$, and $S_I =  \{(x \in \cap_j S_j) < \theta_U\}$.
The estimator for $n_{P, I}$ is $\romTS_{P, I} := |P(S_I)| / \theta_U$. 

It is not difficult to see that $\romTS_{P, I}$ is exactly equal to $\romTS_{P', U}$, where $P'$ is the property that evaluates to 1 on an identifier if and only if
the identifier satisfies $P$ and is also in $I$. Since the latter estimator was already shown to be unbiased with variance bounded as per Theorem \ref{thm:variance},
$\romTS_{P, I}$ satisfies the same properties. 

\eat{\subsection{Practical Notes}\label{sec:practical-notes}

It might appear from Algorithm~\ref{code:framework} that the Theta-Sketch Framework makes two passes over the 
input stream: one to compute $\theta$, and another one to compute $S$.  Actual instantiations of the framework (such as when
the base algorithms are KMV, Adaptive Sampling, pKMV, or the Alpha Algorithm) can
perform both operations in a single pass. 

Also, it is highly advantageous for real systems to be able to compute and save the results of 
sub-expressions of larger set expressions as soon as the necessary data becomes available. Because 
theta-sketches remember the sampling level via $\theta$, it is possible to simply drop
samples from $S$ that don't satisfy a set expression.\footnote{By contrast, the AKMV system
of \cite{beyer2009distinct} must keep those ``non-satisfying samples'' as placeholders, so it marks them using a
extra vector of $k$ bits that is included in each sketch.} An implementation of the
Theta-Sketch Framework can enable this ``sketches from set expressions''
feature by providing the following function SetOp2(). 
Arbitrary set expressions can then be handled in the obvious way using evaluation trees.

\begin{algorithmic}[1]{\small
\STATE Function SetOp2 ($(\theta_1,S_1)$, $(\theta_2,S_2)$, {\bf op} $\in \{\cup,\cap,\setminus\}$ )
\STATE $\theta_U \leftarrow \min (\theta_1, \theta_2)$.
\STATE $S_r \leftarrow \{ (x \in (S_1\; \mathrm{\bf op}\; S_2)) < \theta_U \}$.
\RETURN $(\theta_U,S_r)$.}
\end{algorithmic}
}
\subsection{The pKMV Variant of KMV}
\label{sec:pkmv}

\label{app:pkmv}
\noindent \textbf{Motivation.}
An internet company involved in online advertising typically faces
some version of the following problem: there is a huge stream of
events representing visits of users to web pages, and a huge
number of relevant ``profiles'', each defined by the combination of a predicate
on users and a predicate on web pages. On behalf of advertisers, the
internet company must keep track of the count of distinct users
who generate events that match each profile. The distribution (over profiles)
of these counts typically is highly skewed and covers a huge dynamic
range, from hundreds of millions down to just a few.

Because the summed cardinalities of all profiles is huge, the brute
force technique (of maintaining, for each profile, a hash table of distinct user ids) would use an impractical amount of space.  
A more sophisticated approach would be to run $\romIKMV$,
treating each profile as separate stream $A_i$. This effectively replaces each hash table in the brute force approach with a KMV sketch.
The problem with $\romIKMV$ in this setting is that, while KMV does avoid storing the entire data stream for streams containing more than $k$ distinct identifiers, KMV
produces no space savings for streams shorter than $k$. 
Because the vast majority of profiles contain only a few users, replacing the hash tables in the brute force approach by KMV
sketches might still use an impractical amount of space.

On the other hand, fixed-threshold sampling with $\theta = p$ for a suitable sampling rate $p$, would
always result in an expected factor $1/p$ saving in space, relative to storing the entire input stream. 
However, this method may result in too large a sample rate for long streams (i.e., for profiles satisfied by many users), also resulting in an impractical amount of space.

\noindent \textbf{The $\pIKMV$ algorithm.} In this scenario, the hybrid Threshold Choosing Function $T(k, A, h) =                                                                      
\min (m_{k+1}, p)$ can be a useful compromise, as it ensures that even short streams get downsampled by a factor of $p$,
while long streams produce at most $k$ samples. 
While it is possible to
prove that this TCF satisfies $1$-Goodness via a direct case analysis, the property
can also established by an easier argument: Consider a hypothetical
computation in which the $\thetaunion$ procedure is used to combine two
sketches of the same input stream: one constructed by KMV with
parameter $k$, and one constructed by fixed-threshold sampling with
parameter $p$.  Clearly, this computation outputs $\theta = \min                                                                           
(m_{k+1}, p)$. Also, since KMV and fixed-threshold sampling both
satisfy $1$-Goodness, and $\thetaunion$ preserves $1$-Goodness (cf. Theorem \ref{thm:unbiased}), $T$ also satisfies $1$-Goodness.

It is easy to see that Condition \ref{mildcondition} applies to $T(k, A, h) = \min (m_{k+1}, p)$ as well. Indeed,
$T$ is clearly order-insensitive, so it suffices to show that adding an additional identifier to the stream cannot
increase the resulting threshold. Since $p$ never changes,
the only way that adding another distinct item to the stream could increase the threshold would be
by increasing $m_{k+1}$. However, that cannot happen.

\section{Alpha Algorithm}
\label{sec:alpha}
\subsection{Motivation and Comparison to Prior Art}
Section \ref{sec:framework}'s
theoretical results 
are strong because they cover such a wide class of base sampling algorithms. In fact, $1$-Goodness even 
covers base algorithms that lack certain traditional properties such as invariance to permutations of the 
input, and uniform random sampling of the input. We are now going to take advantage of these strong
theoretical results for the Theta Sketch Framework by devising a novel base sampling algorithm that lacks those traditional
properties, but still satisfies $1$-Goodness. 
Our main purpose for describing our Alpha Algorithm in detail
is to exhibit the generality of the Theta-Sketch Framework. Nonetheless the Alpha Algorithm does have the following advantages relative to HLL, KMV, and Adaptive Sampling.

\medskip
\noindent \textbf{Advantages over HLL.} 
Unlike HLL,
the Alpha Algorithm provides unbiased estimates for
$\distinct_P$ queries for non-trivial predicates $P$. Also, when instantiating the Theta-Sketch Framework
via the Alpha Algorithm in the multi-stream setting, the error behavior scales better than HLL for general set operations (cf. Section \ref{sec:priorwork}). Finally,
because the Alpha Algorithm computes a sample, its output is human-interpretable and amenable to post-processing. 

\medskip
\noindent \textbf{Advantages over KMV.} Implementations of KMV must either use a heap data structure or quickselect \cite{quickselect61} to give quick access to the 
$k\!+\!1^{\text{st}}$
smallest unique
hash value seen so far. The heap-based implementation yields $O(\log k)$ update time, and quickselect, while achieving $O(1)$ update time, hides a large constant factor in the Big-Oh notation
(cf. Section \ref{sec:priorwork}).  
The Alpha Algorithm avoids the need for a heap or quickselect, yielding superior practical performance. 

\medskip
\noindent \textbf{Advantages over Adaptive Sampling.} The accuracy of Adaptive Sampling oscillates as $n_A$ increases. The Alpha Algorithm avoids this
behavior. 

\medskip
The remainder of this section provides a detailed analysis of the Alpha Algorithm. In particular, we show that it satisfies $1$-Goodness, and 
we give quantitative bounds on its variance in the single-stream setting. 
Later (see Section \ref{app:experiments}), we describe experiments showing that, in both the single- and multi-stream settings, the Alpha Algorithm
achieves a novel tradeoff between accuracy, space usage, update speed, and applicability. 

\paragraph{Detailed Section Roadmap.} 
Section \ref{sec:alpha-tcf} describes the 
threshold choosing function AlphaTCF that creates the instantiation
of the Theta Sketch Framework whose base algorithm we refer to as the Alpha Algorithm.
Section~\ref{sec:alpha-sat} establishes that AlphaTCF satisfies $1$-Goodness, implying, via 
Theorem~\ref{thm-key-property-implies-unbiased-for-single-streams}
that EstimateOnSubPopulation() is unbiased on single streams and on unions and intersections of streams in the framework
instantiation created by plugging in AlphaTCF.
 Section~\ref{sec:alpha-space} bounds the space usage of the Alpha Algorithm, as well as its variance in
 the single-stream setting. 
Section \ref{multistream} discusses
the algorithm's variance in the multistream setting. Finally, 
Section \ref{sec:hipsec} describes the HIP estimator derived from the Alpha Algorithm (see Section \ref{app:priorwork} for an introduction to HIP estimators).

\label{app:alpha}

\subsection{AlphaTCF}\label{sec:alpha-tcf}

Algorithm~\ref{code:alpha-tcf} describes the threshold choosing function AlphaTCF. AlphaTCF can be
viewed as a tightly interleaved combination of two different processes. One process uses the set $D$ to remove
duplicate items from the raw input stream; the other process uses \edit{uses a technique similar to} Approximate Counting \cite{morris1978counting}
to estimate the number of items in the de-duped stream created by the first process. In addition, the second
process maintains and frequently reduces a threshold $\theta = \alpha^i$ that is used by the first process to identify 
hash values that {\em cannot} be members of $S$, and therefore don't need to be placed in the de-duping set $D$,
thus limiting the growth of that set.

If the set $D$ is implemented using a standard dynamically-resized
hash table, then well-known results imply that the amortized cost\footnote{Recent
theoretical results imply that the update time can be made worst-case $O(1)$ \cite{deamortizedcuckoo1, deamortizedcuckoo2}.}
of processing each stream element is $O(1)$, and the space occupied by
the hash table is $O(|D|)$, which grows logarithmically with $n$.

However, there is a simple optimized implementation of the Alpha
Algorithm, based on Cuckoo Hashing, that implicitly, and at zero cost,
deletes all members of $D$ that are not less that $\theta$, and therefore
are not members of $S$ (see Section \ref{sec:experiments}). This does not affect correctness, because those 
deleted members will not be needed for future de-duping tests of hash
values that will all be less than $\theta$. Furthermore, in Theorem~\ref{space-theorem} below, it
is proved that $|S|$ is tightly concentrated around $k$. Hence, the
space usage of this optimized implementation is $O(k)$ with probability $1-o(1)$.

\begin{algorithm}[t]
\caption{The Alpha Algorithm's Threshold Choosing Function}\label{code:alpha-tcf}
\begin{algorithmic}[1]{\footnotesize
\STATE Function AlphaTCF (target size $k$, stream $A$, hash function $h$)
\STATE $\alpha \leftarrow k/(k+1)$.
\STATE prefix$(h(A)) \leftarrow$ shortest prefix of $h(A)$ containing exactly
       $k$ unique hash values.
\STATE suffix$(h(A)) \leftarrow$ the corresponding suffix.
\STATE $D \leftarrow$ the set of unique hash values in prefix$(h(A))$.  
\STATE $i \leftarrow 0$.
\FORALL{$x \in \mathrm{suffix}(h(A))$}
\IF{$x < \alpha^i$}  \label{code:alpha-tcf-line-8}
\IF{$x \not \in D$}  \label{code:alpha-tcf-line-9}
\STATE $i \leftarrow i + 1$. \label{code:alpha-tcf-line-10}
\STATE $D \leftarrow D \cup \{x\}$.
\ENDIF
\ENDIF
\ENDFOR
\RETURN $\theta \leftarrow \alpha^{i}$. \label{code:alpha-tcf-line-15}
}\end{algorithmic}
\end{algorithm}


\subsection{AlphaTCF Satisfies $1$-Goodness}\label{sec:alpha-sat}

We will now prove that AlphaTCF satisfies $1$-Goodness.

\begin{theorem}\label{thm:alpha-sat}
If $T(X^{n_A}) = $ AlphaTCF, then every fix-all-but-one projection $\txn$ of $T(X^{n_A})$ satisfies $1$-Goodness.
\end{theorem}
\begin{proof}
Fix the number of distinct identifiers $n_A$ in $A$. Consider any identifier $\ell$ appearing in the stream,
and let $x=h(\ell)$ be its hash value. Fix the hash values of all other elements of the sequence of values $\xnml$.
We need to exhibit a threshold $F$ such that $x < F$ implies $T_\ell[\xnml](x_\ell)(x) = F$ and $x \ge F$ implies $T_\ell[\xnml](x) \le x$.

First, if $x$ lies in one of the first $k+1$ positions in the stream, then $T_\ell[\xnml](x)$ is a constant independent of $x$; in
this case, $F$ can be set to that constant.


Now for the main case, suppose that $\ell$ does not lie in one of the first $k+1$ positions of the stream.
Consider a subdivision of the hashed stream into the initial segment preceding $x=h(\ell)$, 
then $x$ itself, then the final segment that follows $x$. Because all hash values besides $x$ are fixed in $\xnml$,
during the initial segment, there is a specific number $a$ of times that 
$\theta$ is decreased. When $x$ is processed, $\theta$ is decreased either zero or one times,
depending on whether $x < \alpha^a$.
Then, during the final segment, $\theta$ will be decreased a certain number of additional times, where 
this number depends on whether $x < \alpha^a$. Let $b$ denote the number of additional times $\theta$
is decreased if $x < \alpha^a$, and $c$ the number of additional times $\theta$ is decreased otherwise.
This analysis is summarized in the following table:

\begin{center}
\begin{tabular}{|c|c|c|}
\hline
Rule & Condition on $x$ & Final value of $\theta$ \\
\hline
\edit{L} & $x < \alpha^a$ & $\alpha^{a + b + 1}$ \\
\hline
\edit{G} & $x \ge \alpha^a$ & $\alpha^{a + c + 0}$ \\
\hline
\end{tabular}
\end{center}

We prove the theorem using the threshold $F = \alpha^{a + b + 1}$. We note that
$F = \alpha^{a + b + 1} < \alpha^a$, so $F$ and $\alpha^a$ divide the range of $x$
into three disjoint intervals, creating three cases that need to be considered.

Case 1: $x < F < \alpha^a$. In this case, because $x < F$, we need to show that $T_\ell[\xnml](x) = F$.
By Rule \edit{L}, $T_\ell[\xnml](x) = \alpha^{a + b + 1} = F$.

Case 2: $F \le x < \alpha^a$. Because $x \ge F$, we need to show that $T_\ell[\xnml](x) \le x$. 
By Rule \edit{L}, $T_\ell[\xnml](x) = \alpha^{a + b + 1} = F \le x$.

Case 3: $F < \alpha^a \le x$. Because $x \ge F$, we need to show that $T_\ell[\xnml](x) \le x$. 
By Rule \edit{G}, $T_\ell[\xnml](x) = \alpha^{a + c + 0} \le \alpha^a \le x$.
\end{proof}

\subsection{Analysis of Alpha Algorithm on Single Streams}\label{sec:alpha-space}
The following two theorems show that the Alpha Algorithm's space usage
and single-stream estimation accuracy are quite similar to those of KMV.
That means that it is safe to use the Alpha Algorithm as a drop-in replacement
for KMV in a sketching-based big-data system, which then allows the system to benefit 
from the Alpha Algorithm's low update cost. See the Experiments in Section~\ref{app:experiments}.

\noindent \textbf{Random Variables.} When Line~\ref{code:alpha-tcf-line-15} of
Algorithm~\ref{code:alpha-tcf} is reached after processing a randomly hashed
stream, the program variable $i$ is governed by a random variable $\mathcal{I}$.
Similarly, when Line~\ref{code:framework-line-3} of Algorithm~\ref{code:framework}
is subsequently reached, the cardinality of the set $S$ is governed by a random 
variable $\mathcal{S}$. The following two theorems characterize the distributions
of $\mathcal{S}$ and of the Theta Sketch Framework's estimator $\mathcal{S}/(\alpha^\mathcal{I})$.
Specifically, Theorem \ref{space-theorem} shows that the number of elements sampled by the Alpha Algorithm is tightly concentrated around $k$,
and hence its space usage is concentrated around that of KMV.
Theorem \ref{thm:alpha-basic-variance} shows that the variance of the estimate returned by the Alpha Algorithm
is very close to that of KMV.
Their proofs are rather involved, and are deferred to Appendices \ref{appendix-proof-of-space-theorem} and
\ref{appendix-proof-of-alpha-basic-variance} respectively.

\begin{theorem}\label{space-theorem}
Let $\mathcal{S}$ denote the cardinality of the set $S$ computed by the Alpha Algorithm's Threshold Choosing Function (Algorithm~\ref{code:alpha-tcf}). Then:
\begin{align} 
\ept(\mathcal{S}) = & \; k. \\
\sigma^2(\mathcal{S}) < & \; \frac{k}{2} + \frac{1}{4}.
\end{align}
\end{theorem}


\begin{theorem}\label{thm:alpha-basic-variance}
Let $\mathcal{S}$ denote the cardinality of the set $S$ computed by the Alpha Algorithm's Threshold Choosing Function (Algorithm~\ref{code:alpha-tcf}). Then:
\begin{align}
\sigma^2(\mathcal{S}/(\alpha^\mathcal{I})) = & \frac{(2k+1)n_A^2 - (k^2+k)(2n_A-1)- n_A}{2k^2} \\
                         < & \frac{n_A^2}{k - \frac{1}{2}}.
\end{align}
\end{theorem}

\subsection{Variance of the Alpha Algorithm in the Multi-Stream Setting}
\label{multistream}
Unfortunately, the Alpha Algorithm does not satisfies monotonicity (Condition \ref{mildcondition}) in general,
and hence Theorem \ref{thm:variance} does not immediately imply variance bounds
in the multi-stream setting. In fact, we have identified contrived examples
in the multi-stream setting on which the variance of the Theta-Sketch Framework when instantiated with the TCF of the Alpha Algorithm
is slightly larger than the hypothetical estimator obtained by running the Alpha Algorithm on the concatenated stream $A_1 \circ \dots A_m$ (the worst-case
setting appears to be when $A_1 \dots A_m$ are all permutations of each other). 

However, we show in this section that the Alpha Algorithm does satisfy monotonicity under the promise that all constituent streams are pairwise disjoint.
This implies the variance guarantees of Theorem \ref{thm:variance} do apply to the Alpha Algorithm under the promise that $A_1, \dots, A_m$ are pairwise disjoint.
Our experiments in Section~\ref{app:webscope-experiment} suggest that, in practice, the \edit{normalized} variance of the Alpha Algorithm in the multi-stream setting is not much larger than in the pairwise disjoint case.

\begin{theorem}
The TCF computed by the Alpha Algorithm satisfies Condition \ref{mildcondition} under the promise that the streams $A_1, A_2, A_3$ appearing
in Condition \ref{mildcondition} are pairwise disjoint.
\end{theorem}
\begin{proof}
Inspection of Algorithm~\ref{code:alpha-tcf} shows that the Alpha Algorithm never increases $\theta$ while processing a stream.
Therefore, processing $A_3$ after $A_2$ cannot increase $\theta$ above the value that it had at the end of processing
$A_2$. Hence, it will suffice to prove that $T(A_2) \ge T(A_1 \circ A_2)$. Referring to Line~\ref{code:alpha-tcf-line-15} of the pseudocode, we see
that $\theta = \alpha^I$, where $I$ is the final value of the program variable $i$, so it suffices to
prove that $I(A_2) \le I(A_1 \circ A_2)$. 

We will compare two execution paths of the Alpha Algorithm. The first path results from processing
$A_2$ by itself. The second path results from processing $A_1 \circ A_2$. We will now
index the sequence of hash values of $h(A_1 \circ A_2)$ in a special way: $x_0$ will be the
first hash value that reaches Line~\ref{code:alpha-tcf-line-8} of the pseudocode during the first execution path (where $A_2$ is
processed by itself). Elements of $h(A_1 \circ A_2)$ that follow $x_0$ will be numbered $x_1, x_2, \ldots$, while elements
of $h(A_1 \circ A_2)$ that precede $x_0$ will be numbered $\dots, x_{-2}, x_{-1}$. We remark that
the boundary between negative and positive indices does not coincide with the boundary between $A_1$ and $A_2$.

For $j \ge 0$, let
$I(j)$ denote the value of the program variable $i$ immediately before processing the hash value
$x_j$ on the first execution path ($A_2$ alone), and let $I'(j)$ denote the same quantity for the second
execution path ($A_1 \circ A_2$). We will prove by induction that for all $j \geq 0$, $I(j) \le I'(j)$.
The base case is trivial: by construction of our indexing scheme, at position $0$, execution
path one has had no opportunities yet to increment $i$, while execution path two might have had some opportunities to increment $i$.
Hence $I(0) = 0$ while $I'(0) \ge 0$. 

Now for the induction step. At position $j$, $I(j) \le I'(j)$, and the two values of $i$
are both integers, so the only possible way for $I(j+1) > I'(j+1)$ to occur would be for $I(j) = I'(j)$, and for the
tests at
Line~\ref{code:alpha-tcf-line-8}
and
Line~\ref{code:alpha-tcf-line-9} of the pseudocode to both {\em pass} on the first execution path, while at least one of them
{\em fails} on the second execution path. However, the test in Line~\ref{code:alpha-tcf-line-8} must have the same outcome for both paths,
since they are comparing the same hash value $x_j$ against the same threshold $\alpha^i = \alpha^{i'}$. Also, given the
assumption that $A_1$ and $A_2$ are disjoint, the ``novelty test'' in Line~\ref{code:alpha-tcf-line-9} is
determined solely by novelty within
$A_2$. Hence, it must have the same outcome on both paths. We conclude that it is impossible for $i$ to be incremented on the
first path but not on the second path, so $I(j+1) > I'(j+1)$ is impossible.
\end{proof}

\subsection{HIP estimator} 
\label{sec:hipsec} 
For single streams, the HIP estimator (see Section \ref{app:priorwork} for an introduction to HIP estimators) derived from the Alpha 
Algorithm turns out to equal $k/\alpha^i$.
This estimator does not involve the size of the sample set $S$, and is therefore not the same thing
as the estimator $|S|/\alpha^i$ derived by instantiating the Theta-Sketch Framework with 
the Alpha Algorithm.
The following theorem shows that the variance bound of the HIP estimator guaranteed by
Theorem \ref{thm:hip} is smaller than the variance bound for the
vanilla Alpha Algorithm (cf. Theorem \ref{thm:alpha-basic-variance}) by a
factor of 2. It 
can be proved by using the analysis of Approximate Counting
of \cite{morris1978counting,flajolet1985approximate},
or by using the analysis of HIP estimators of \cite{cohennew, ting}. To 
keep the paper self-contained, Appendix \ref{app:hip} contains
a proof of this result that utilizes several immediate results developed in the proof of
Theorem \ref{thm:alpha-basic-variance}.

\begin{theorem} \label{thm:hip} Let $n_A$ denote the number of distinct elements in stream $A$.
If $\alpha^i = \mathrm{AlphaTCF}(k,A,h)$, then: 
\begin{align*}
\ept(k/\alpha^i) = & n_A, \\
\sigma^2(k/\alpha^i) = & \frac{n_A^2-2n_Ak+k^2-n_A+k}{2k} < \frac{n_A^2}{2k}, \\
\mathrm{S.E.}(k/\alpha^i) < & 0.708 / \sqrt{k}.
\end{align*}
\end{theorem}

\section{Experiments}
\subsection{Single-Stream Experiments Using Synthetic Data}
\label{sec:experiments}
\label{app:experiments}


In this section we describe experiments using synthetic data
showing that implementations of KMV, Adaptive Sampling, and the Alpha Algorithm can provide different tradeoffs
between time, space, and accuracy in the single-stream setting.
All three implementations take advantage of a version of cuckoo hashing
that treats as empty all slots containing hash values that are not
less than the current value of $\theta$.

The code for our streaming implementation of the Alpha Algorithm
closely resembles the pseudocode presented as Algorithm~\ref{code:alpha-tcf}.
The de-duping set $D$ is stored in a cuckoo hash table that uses the just-mentioned self-cleaning
trick. Hence $D$ is in fact {\em always} equal to $S$, with no extra
work needed in the form of table rebuilds or explicit delete operations.

Our implementation of Adaptive Sampling uses the same self-cleaning
hash tables, but has a different rule for reducing $\theta$: multiply by 1/2
each time $|S|$ reaches a pre-specified limit. Again, no delete operations
or table rebuilds are needed, but this program needs to scan the table
after each reduction in $\theta$ to discover the current size of $|S|$.

Finally, our implementation of KMV again uses the same self-cleaning
hash tables, but it also uses a heap to keep track of the current
value of $\theta = m_{k+1}$. Hence it either uses more space than
the other two algorithms, or it suffers from a reduction in
accuracy due to sharing the space budget between the hash table
and the heap. Also, it is slower than the other two algorithms
because it performs heap operations in addition to hash table
operations.

Our experiments compare the speed and accuracy on single streams
of these implementations of the three algorithms. 
Accuracy was evaluated using the metric $\sqrt{\mathrm{meanSquaredError}}/n_A$,
measured during the course of 1 million runs of each algorithm.
We employed two different sets of experimental conditions. 

First, we compare under ``equal-$k$'' conditions, in which all three
algorithms aim for $|S| = t/2$, where $t=2^{16}$ denotes the size of
the hash table. Adaptive Sampling is configured to oscillate between
roughly $|S| = (1/3)t$ and $|S| = (2/3)t$.  We remark that KMV
consumes more space than the other two algorithms under these
conditions because of its heap.

Second, we compare under ``equal-space'' conditions reflective of a
live streaming system that needs to limit the amount of memory
consumed by each sketch data structure. Under these conditions, KMV is
forced to devote half of its space budget to the heap, while both
Adaptive Sampling and the Alpha Algorithm are free to employ
parameters that cause their hash tables to run at occupancy levels
well over 1/2.  In detail, for KMV $|S|=(2/5)t$, for the Alpha
Algorithm $|S|=(4/5)t$, while Adaptive Sampling oscillates between
roughly $|S| = (2/5)t$ and $|S| = (4/5)t$.

Experimental results are plotted in Figure~\ref{fig:base-algo-tradeoffs}. Two things
are obvious. First, the heap-based implementation of KMV
is much slower than the other two algorithms. Second, the error curves
of Adaptive Sampling have a strongly oscillating shape that can be 
undesirable in practice.

Under the equal-$k$ conditions, the error curves of KMV and the Alpha
Algorithm are so similar that they cannot be distinguished from each
other in the plot. However, under the equal space conditions, the
Alpha Algorithm's ability to operate at a high, steady occupancy level (of
the hash table) causes its error to be the lowest of the three
algorithms. This high, steady occupancy level also causes the Alpha Algorithm
to be slightly slower than Adaptive sampling under these conditions,
even though the latter needs to re-scan the table periodically, 
while the Alpha Algorithm does not.

\begin{figure}
\begin{center}
\includegraphics[width=0.45\linewidth]{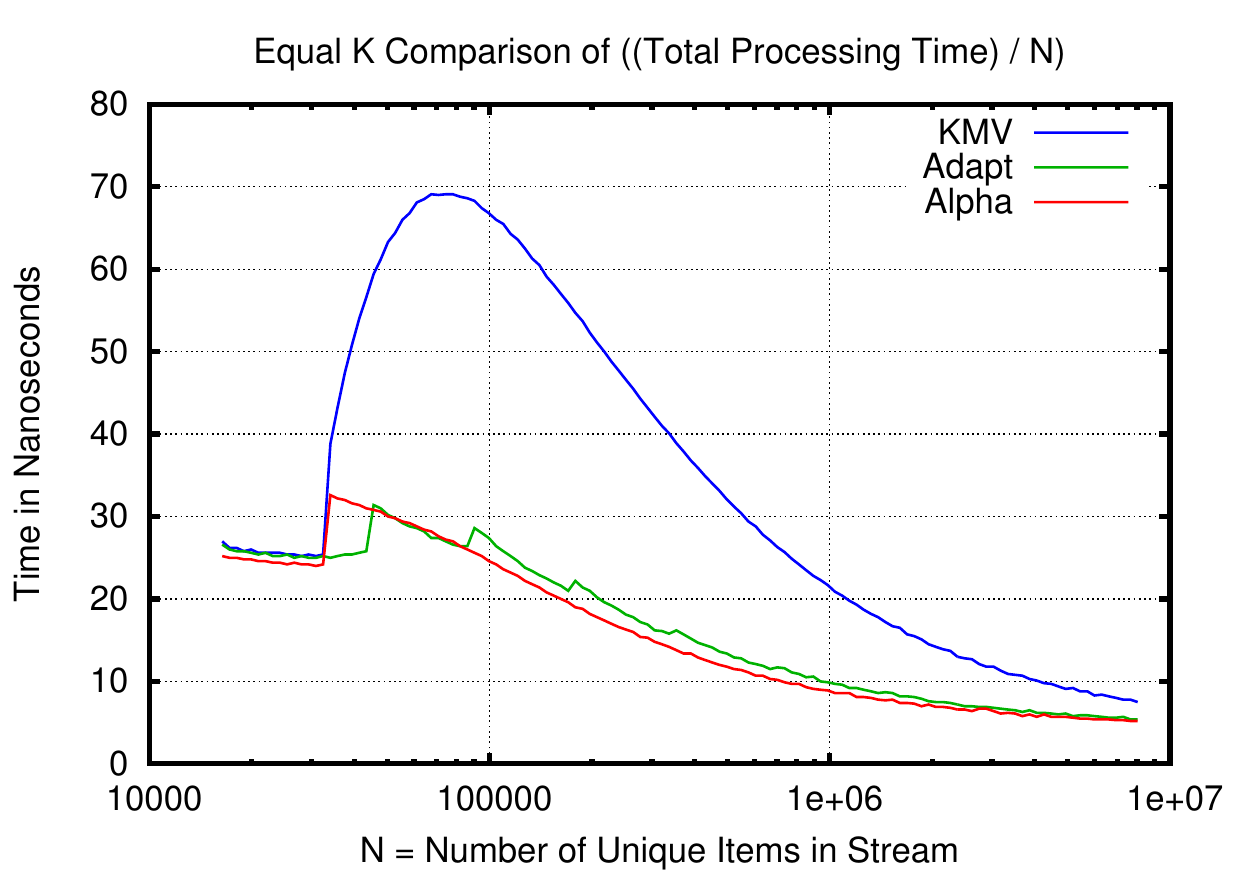}\quad
\includegraphics[width=0.45\linewidth]{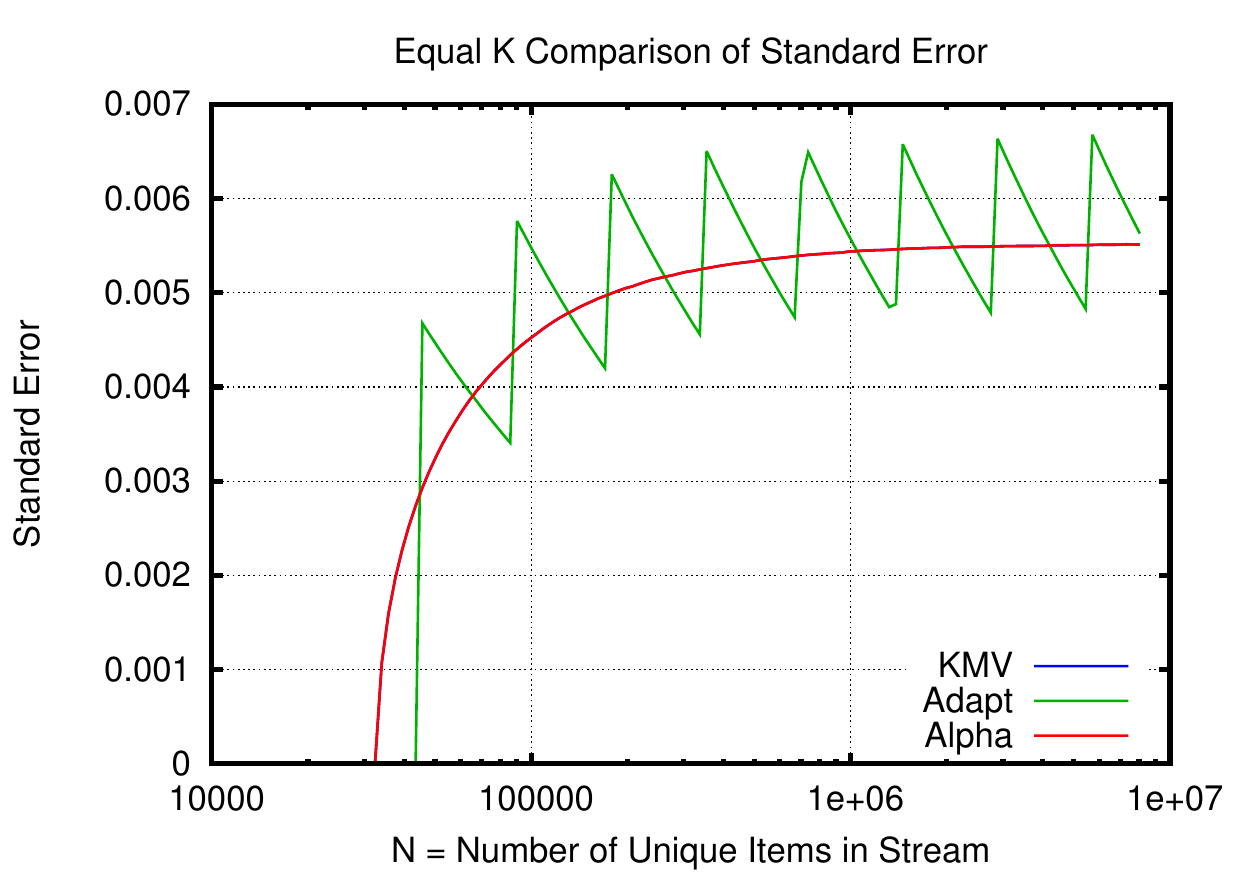} \\
\includegraphics[width=0.45\linewidth]{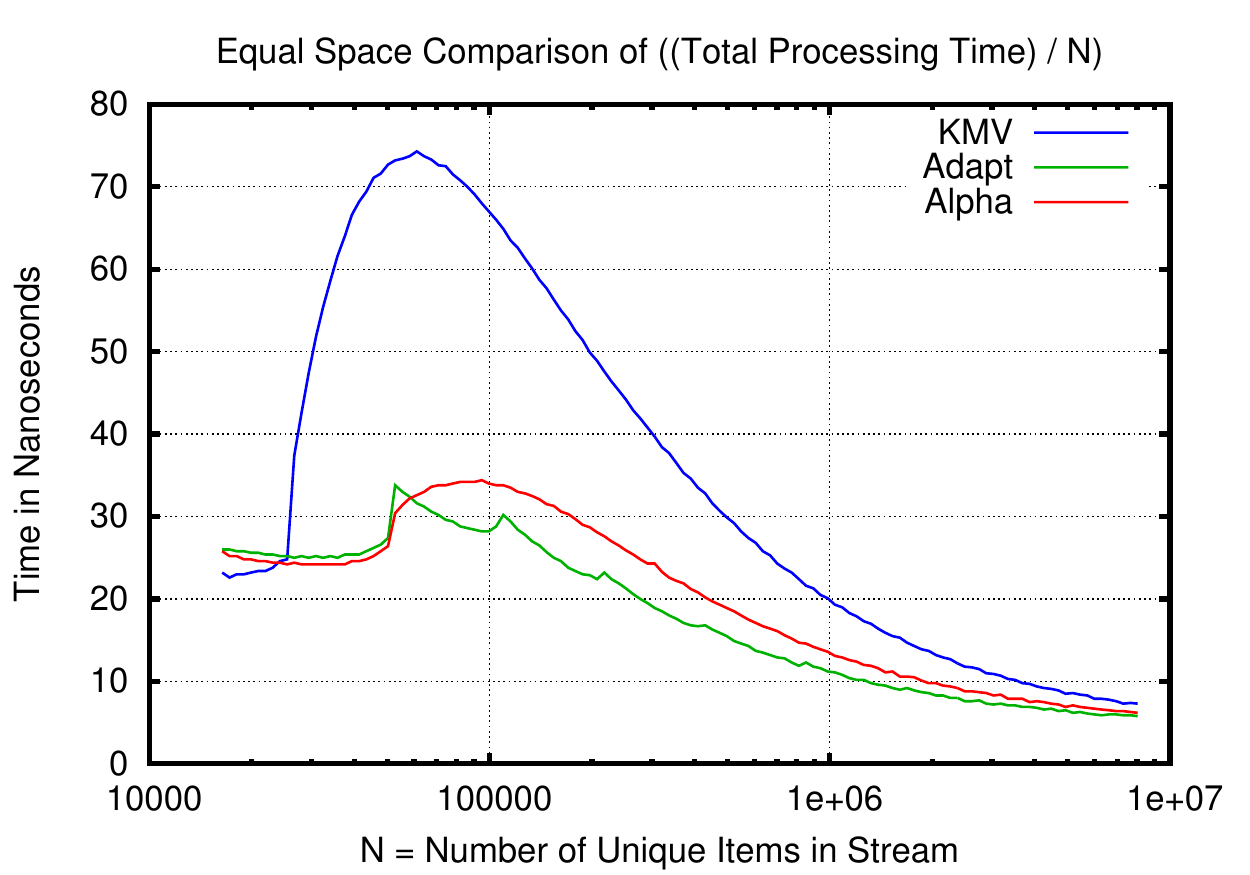}\quad
\includegraphics[width=0.45\linewidth]{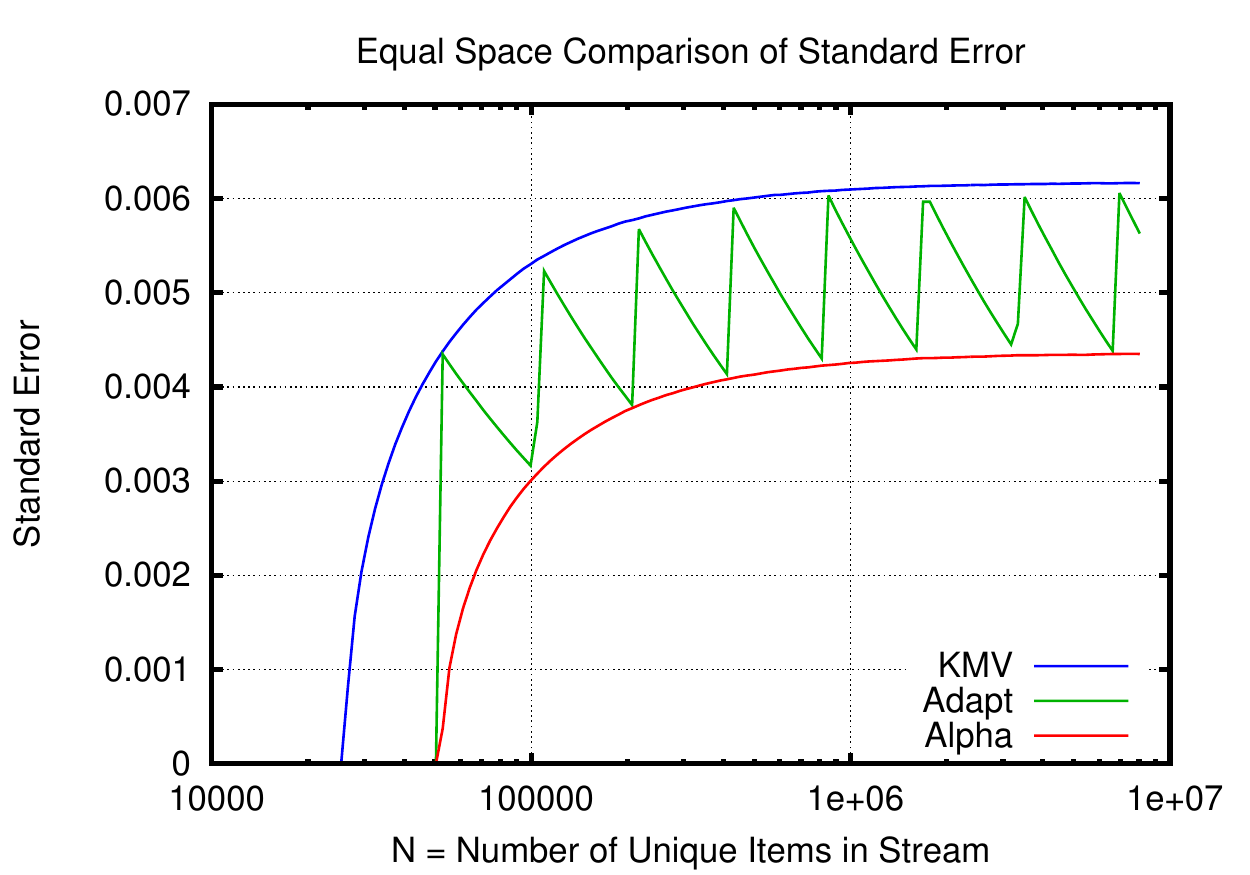} \\
\end{center}
\caption{These plots illustrate the low stream processing cost and non-oscillating
error curves of the Alpha Algorithm.}
\label{fig:base-algo-tradeoffs}
\end{figure}

\subsection{A Multi-Stream Experiment Using Real Data}
\label{sec:webscope-experiment}
\label{app:webscope-experiment}
As discussed in Section \ref{sec:applicability}, Theorem~\ref{thm:variance}'s comparative
variance result does not apply to the Alpha Algorithm in general. However, 
we proved in Section \ref{app:alpha} that Theorem \ref{thm:variance} does apply to the Alpha Algorithm when the
input streams are disjoint. In this section we present empirical
evidence suggesting that the Alpha Algorithm ``almost'' satisfies the variance bound of Theorem \ref{thm:variance} on real
data. Recall that Theorem \ref{thm:variance} asserted that $\sigma^2(\hat{n}^U_{P,U}) \le                                                      
\sigma^2(\hat{n}^{A^*}_{P,A^*})$ when the estimates are computed using TCFs satisfying $1$-Goodness and monotonicity. 
Simplifying notation, and switching
from variance to relative error, we will exhibit a scatter plot
comparing $\mathrm{RE}_U(A_1,A_2)$ versus $\mathrm{RE}_{A*}(A_1,A_2)$,
for numerous pairs $(A_1,A_2)$ of sets from a naturally occurring
dataset, using the TCF defined by the Alpha Algorithm. This scatter plot will show that only a tiny fraction of the
pairs violates the bound asserted in the theorem.

\begin{figure}
\begin{center}
\includegraphics[width=0.8\linewidth]{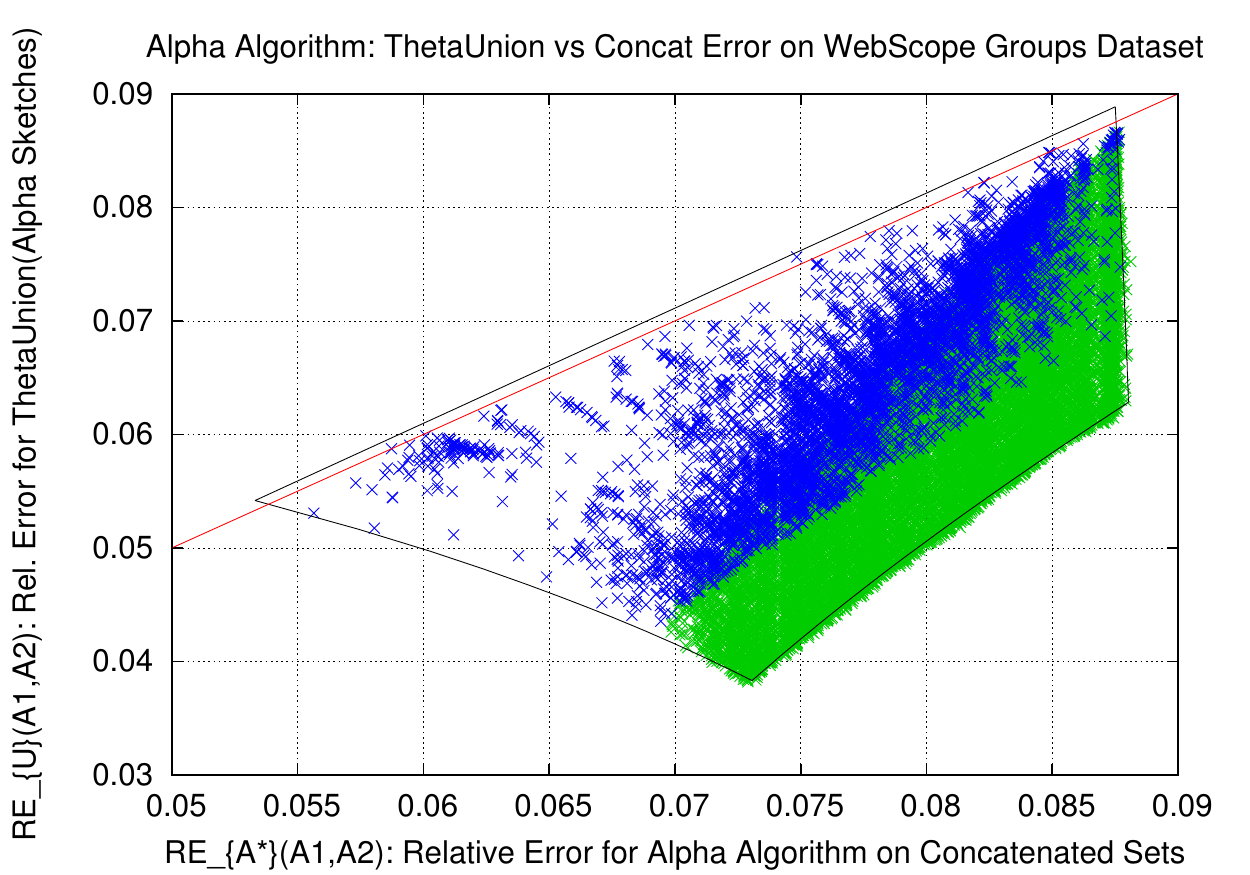}
\end{center}
\caption{Most points are below the red line, showing that the
comparative variance bound of Theorem~\ref{thm:variance} is ``nearly true'' for the Alpha Algorithm on the Webscope
Dataset.}
\label{fig:webscope-groups-scatter}
\end{figure}

\noindent \textbf{WebScope ``Groups'' Dataset.}
This experiment is based on
ydata-ygroups-user-group-membership-graph-v1\_0, a dataset that is
available from the Yahoo Research Alliance Webscope program.  It
contains anonymized and downsampled membership lists for about 640000
Yahoo Groups, circa 2005. Because of the downsampling, there are only
about 1 million members in all. We restricted our attention to the
roughly 10000 groups whose membership lists contained between 201 and
5429 members. Hence there were about 50 million pairs of groups to
consider.  Recalling that the comparative variance theorem applies to the Alpha Algorithm
under the promise that groups are disjoint, we trimmed this set of 50 million
pairs down to 5000 pairs that seemed most likely to violate the 
theorem because they had the highest overlaps as measured by the similarity score 
$\textrm{sim}(A_1,A_2) := (|A_1 \cap A_2| / \min (|A_1|,|A_2|))$.
We also examined another 13000 pairs of groups to fill out the 
scatter plot.                                                                                      

For each of these roughly 18000 pairs of groups, we empirically
measured, by means of 100000 trials with $k$ set to 128,
the values of $\mathrm{RE}_U(A_1,A_2)$ and
$\mathrm{RE}_{A*}(A_1,A_2)$, and plotted them in the scatter plot
appearing as Figure~\ref{fig:webscope-groups-scatter}.
The 5000 high-overlap pairs are plotted in blue, while the other 13000
pairs are plotted in green.  Strikingly, all but 2 of the roughly
18000 points lie on or below below the red line, thus indicating an
outcome that is consistent with the comparative variance result.
Because we included every pair of sets that had large overlap
(as measured by $\textrm{sim}(A_1,A_2)$) we conjecture that all of the
other roughly 50 million pairs of sets also conform to the theorem.



Figure~\ref{fig:webscope-groups-scatter} also includes a heuristic
``bounding box'' plotted as a black quadrilateral.
This bounding box was computed numerically from
several ingredients.  For the $\mathrm{RE}_U(A_1,A_2)$ side of the
computation, we exploited the fact that for any given values of $n$
and $k$, the Alpha Algorithm's exact distribution over $\theta$ values
can be computed by dynamic programming using recurrences similar to
the ones described in \cite{flajolet1985approximate}.  
We also made
the (counter-factual) assumption that three different hash functions are
used to process the input sets $A_1$ and $A_2$, and the output set
$S_U$. This breaks the dependencies which complicate the analysis of the actual
multi-stream instantiation of the Alpha Algorithm, in which only a single hash
function is used during any given run. However, this counter-factual
assumption also means that the resulting bounding
box is not quite accurate. Finally, we did a grid search over
all possible set-size pairs
$201 \le |A_1| \le |A_2| \le 5429$ and all
possible amounts of overlap, and traced out the boundary of the
resulting combinations of computed relative errors.
This boundary (see Figure~\ref{fig:webscope-groups-scatter}) suggests
that the comparative variance theorem is true for nearly all {\em possible} triples
$(|A_1|,|A_2|,|A_1 \cap A_2|)$ where $201 \le |A_1| \le |A_2| \le 5429$ and
$|A_1 \cap A_2| \le A_1$. Moreover, in those relatively few
cases where the theorem is violated, the magnitude of the
violation is small.

\section{Detailed Overview of Prior Work}
\label{app:priorwork}
\subsection{Algorithms for Single Streams}
\noindent \textbf{HLL: HyperLogLog Sketches.}\label{subsubHLL}
HLL is a sketching algorithm for the vanilla \distinct\ problem. It uses a hash function to randomly distribute the elements of a stream $A$ amongst $k$ buckets.
For each bucket $i$, there is a register $b_i$, whose length is $O(\log \log n)$ bits, that essentially
contains the largest number of leading zeros in the hashed value of any stream element sent to that bucket.
For each stream element, this data structure can clearly be updated in $O(1)$ time.
The HLL estimator for $n_A$ which we denote $\romHLL_A $, is a certain non-linear function of the $k$ bucket values $b_i$; 
see \cite{flajolet2008hyperloglog}. It has been proved by \cite{flajolet2008hyperloglog} that, as $n_A \rightarrow \infty$,
$E(\romHLL_A) \rightarrow\; n_A$, and $\sigma^2(\romHLL_A) \rightarrow\; 1.04 (n_A^2/k)$. 

Unlike the KMV and Adaptive Sampling algorithms described below, it is not known how to extend the HLL sketch to estimate $n_{P, A}$ for general properties $P$ (unless, of course, $P$ is known prior to stream
processing). 
Qualitatively, the reason that HLL cannot estimate $n_{P, A}$
is that, unlike the other algorithms, HLL does not maintain any kind of sample of identifiers from the stream.

\medskip
\noindent \textbf{KMV: K'th Minimum Value Sketches.}\label{subsubBKMV}
\label{sec:KMV}
The KMV sketching procedure for estimating \distinctA\ works as follows. While processing an input stream $A$, KMV keeps track of the set $S$ of the $k$ smallest unique hashed values of stream elements.
The update time of a heap-based implementation of KMV is $O(\log k)$.
The KMV estimator for \distinctA\ is
\begin{equation} \label{eq:kmv} \romKMV_A = \; k / m_{k+1},\end{equation} where $m_k$ denotes the $k$'th smallest hash value.
It has been proved by \cite{beyer2009distinct}, \cite{giroire2009order}, and others, that $E(\romKMV_A) = n_A$, and
\begin{align}
\sigma^2(\romKMV_A) = & \frac{n_A^2- k \; n_A}{k-1} < \frac{n_A^2}{k-1}.
\end{align}

Duffield et al. \cite{DuffieldLT07} proposed to change 
the heap-based implementation of priority sampling
to an implementation based on quickselect \cite{quickselect61}. 
The same idea applies to KMV, which is a special case of priority sampling,
and it reduces the sketch update cost from $O(\log k)$ to amortized $O(1)$. 
However, this $O(1)$ has a larger constant factor than
that of competing methods.

The KMV sketching procedure can be extended to estimate $n_{P, A}$ for any property $P \subseteq [n]$, as explained below.
To accomplish this, the KMV sketch must keep not just the $k$ smallest unique hash values that have been observed in the stream, but also 
the actual item identifiers corresponding to the hash values.\footnote{Technically, the sketch need not store the hash values if it stores the corresponding identifiers. 
Nonetheless, storing the hash values is often desirable in practice, to avoid the need to repeatedly evaluate the hash function.}
This allows the algorithm to determine which of the items in the sample
satisfy the property $P$, even when $P$ is not known until query time.

Motivated by the identity $n_{P, A} = n_A \cdot (n_{P, A} / n_A)$, the quantity $\romKMV_A \cdot est(n_{P, A}/n_A)$ is a plausible
estimate of $n_{P, A}$, for any sufficiently accurate estimate $est(n_{P, A}/n_A)$  of $n_{P, A}/n_A$.
Let $S_A$ denote the $k$ smallest unique hashed values in $A$, and recall (cf. Section \ref{sec:hashprelims}) that $P(S_A)$ denotes
the subset of hash values in $S_A$ whose corresponding identifiers in $[n]$ satisfy the predicate $P$ (the reason we require
the sketch to store the actual identifiers that hashed to each value is to allow $S_A$ to be determined from the sketch). 
Then the fraction $|P(S_A)| / |S_A|$ can serve as the desired estimate of the fraction $n_{P, A}/n_A$. 
Essentially because $S_A$ is a uniform random sample of $A$, it can be proved
that the estimate $\romKMV_{P, A} = \romKMV_A \cdot |P(S_A)| / |S_A|$ of $n_{P, A}$ 
is unbiased, and has the following variance:\footnote{\cite{beyer2009distinct} analyzed the closely related estimator
$\romKMV'_{P, A} = \romKMV_A \cdot |P(S'_A)| / |S'_A|$, where $S'_A = S_A \cup \{m_{k+1}\}$,
proving unbiasedness and 
deriving the variance $\sigma^2(\romKMV'_{P, A}) = 
(n_{P, A}((k+1) \; n_A - (k+1)^2 - n_A + k+1 + n_{P, A}))\;\;/\;\;((k+1)(k-1))
$.}
\begin{align}
\sigma^2(\romKMV_{P, A}) = & \frac{n_{P, A} (n_A - k)}{k-1} <
\frac{n_{P, A} \; n_A}{k-1}. 
\end{align}



\medskip \noindent \textbf{Adaptive Sampling.}\label{subsubAdapt}
Adaptive Sampling maintains a sampling level $i \ge 0$, and the set $S$
of all hash values less than $2^{-i}$; whenever $|S|$ exceeds a pre-specified
size limit, $i$ is incremented and $S$ is scanned discarding any hash value
that is now too big. Because a simple scan is cheaper than running
quickselect, an implementation of this scheme can be cheaper than
KMV. The estimator of $n_A$ is $\romAdapt_A = \; |S| /
2^{-i}$. It has been proved by \cite{flajolet1990adaptive} that this
estimator is unbiased, and that $\sigma^2(\romAdapt_A) \approx 1.44
(n_A^2/(k-1))$, where the approximation sign hides oscillations caused
by the periodic culling of $S$.  
Like KMV, Adaptive Sampling can be extended to estimate $n_{P, A}$ for any property $P$, via 
$\romAdapt_{P, A} = \romAdapt_A \cdot |P(S_A)| / |S_A|$. Note that, just as for KMV, this extension requires storing not just the hash values in $S$,
but also the actual identifiers corresponding to each hash value.

Although the stream processing speed of Adaptive Sampling is excellent, 
the fact that its accuracy oscillates as $n_A$ increases 
is a shortcoming of the method. 

\subsection{Algorithms for Set Operations on Multiple Streams}
\noindent \textbf{HLL Sketches for Multiple Streams.}\label{subsubHLL2}
 \begin{itemize}
 \item \textbf{Set Union.} A sketch of $U$ can be constructed from $m$ HLL sketches of the $A_j$'s by taking
the maximum of the $m$ register values for each of the $k$ buckets. 
The resulting sketch is identical to an HLL sketch constructed directly from $U$, so
$E(\romHLL_U) \rightarrow\; n_U$, and
$\sigma^2(\romHLL_U) \rightarrow\; 1.04 (n_U^2/k)$.

\item \textbf{Set Intersection.} Given constituent streams $A_1, \dots, A_m$, the HLL scheme can be extended via the Inclusion/Exclusion (IE) rule to estimate \distinct\ 
for various additional set-expressions other than set-union applied to $A_1, \dots, A_m$. 
This approach is awkward for
complicated expressions, but is straightforward for simple expressions. For example, if $m=2$, then
the HLL+IE estimate of $|I| = | A_1 \cap A_2 | $ is $\romHLL_{A_1} + \romHLL_{A_2} - \romHLL_U$.

Unfortunately, the variance of this estimate is approximately
$n_U^2/k$. This is a factor of $n^2_U/n^2_I$ larger than the variance
of roughly $n_I^2/k$ if one could somehow run HLL directly on $I$, and a factor of $n_U/n_I$ worse than the variance achieved by
the $\romIKMV$ algorithm described below. When $n_I \ll n_U$, this penalty
factor overwhelms HLL's fundamentally good accuracy per bit.
\end{itemize}
In summary, the main limitations of HLL are its bad error scaling behavior when dealing with set operations other than set-union, 
as well as the inability to estimate \distinctP\ queries for general properties $P$, even for a single stream $A$. \\

\medskip \noindent \textbf{\IKMV: KMV for Multiple Streams.}\label{subsubBKMV2} \label{sec:IKMV}
 \begin{itemize}
 \item \textbf{Set Union.} 
For any property $P$, there are two natural ways to extend KMV to estimate $n_{P, U}$, 
given a KMV sketch $S_j$ containing the $k+1$ smallest unique hash values for each constituent stream $A_j$. 
The first is to use a ``non-growing'' union rule, and the second is to use a ``growing'' union rule (our term). 

With the non-growing union rule, the sketch 
of $U$ is simply defined to be the set of $k+1$ smallest unique hash values in $\cup_{j=1}^m S_j$.
The resulting sketch is identical to a KMV sketch constructed directly from $U$, so $E(\romKMV_U) = n_U$, and
$\sigma^2(\romKMV_U) < n_U^2/(k-1)$. 
Just as the KMV sketch for a single stream $A$ can be adapted to estimate $n_{P, A}$ for any property $P$, 
this multi-stream variant of KMV can be adapted to provide an estimate $\romKMV_{P,U}$ of
$n_{P, U}$.

The growing union rule
was introduced by Cohen and Kaplan \cite{cohen2009leveraging}.
This rule decreases the variance of estimates
for unions and for other set expressions, but also increases the space
cost of computing those estimates. Throughout, we refer to Cohen and Kaplan's algorithm as $\rommultiKMV$.
\label{subsubIKMV}
For each KMV input sketch $S_j$, let $M_j$ denote that sketch's value of $m_{k+1}$. 
Define $M_U = \min_{j=1}^m M_j$, and $S_U = \{x \in \cup_j S_j \colon x < M_U\}$.
Then $n_U$ is estimated by $\rommultiKMV_U := |S_U|/M_U$, 
and  $n_{P, U}$ is estimated by $\rommultiKMV_{P, U} := \rommultiKMV_U \cdot |P(S_U)| /|S_U| = |P(S_U)| / M_U$.
\cite{cohen2009leveraging} proved that $\rommultiKMV_{P, U}$ is unbiased and has
variance that 
dominates the variance of the ``non-growing'' estimator $\romKMV_{P,U}$:
\begin{align}
\sigma^2(\rommultiKMV_{P, U}) \le & \sigma^2(\romKMV_{P, U}). 
\end{align}

\item \textbf{Set Intersection.} $\rommultiKMV$ can be tweaked in a natural way to handle set intersection and other set operations.
Specifically, as in the set-union case, define $M_U = \min M_j$, and $S_U = \{x \in \cup_j S_j \colon x < M_U\}$. In addition, define $S_I =  \{(x \in \cap_j S_j) < M_U\}$.
The estimator for $n_{P, I}$ is $\rommultiKMV_{P, I} := \rommultiKMV_U \cdot |P(S_I)| /|S_U| = |P(S_I)| / M_U$.
It is not difficult to see that $\rommultiKMV_I$ is exactly equal to $\rommultiKMV_{P', U}$, 
where $P'= P \cap I$ is the property that evaluates to 1 on an identifier if and only if
the identifier satisfies $P$ and is also in $I$. 
Since the latter estimator was already shown to be unbiased with variance bounded as per Equation \eqref{growing-variance-dominates},
$\rommultiKMV_{P, I}$ satisfies the same properties. 
\end{itemize}


\medskip \noindent \textbf{multiAdapt: Adaptive Sampling for Multiple Streams.}\label{subsubAdapt2}
\begin{itemize}
\item \textbf{Set Union.}
Just as with KMV, for any property $P$, there are two natural ways to extend Adaptive Sampling 
to estimate $n_{P, U}$, given an Adaptive Sampling sketch $S_j$ for each constituent stream $A_j$. 
The first is to use a non-growing union rule, and the second is to use a growing union rule. 
 For brevity, we will
only discuss the growing union rule, as proposed by
\cite{gibbons2001estimating}. We refer to this algorithm as $\romIAdapt$. 
Let $(i_j,S_j)$ be the sketch of the
$j$'th input stream $A_j$.  The union sketch constructed from these
sketches is $(i_U = \max i_j$,\; $S_U = \{x \in \cup S_j\colon x < 2^{-i_U}\})$.  
Then $n_U$ is estimated by $\romIAdapt_U := |S_U| / 2^{-i_U}$,
and $n_{P, U}$ is estimated by $\romIAdapt_{P, U} := \romIAdapt_U \cdot |P(S_U)| / |S_U|$.  
\cite{gibbons2001estimating} proved epsilon-delta bounds on the error of the estimator
$\romIAdapt_{P, U}$, but did not
derive expressions for mean or variance. However, $\romIAdapt$ and $\rommultiKMV$ are in fact both special cases
of our Theta-Sketch Framework, and in Section~\ref{sec:framework}
of this paper we will prove (apparently for the first time) that $\romIAdapt_{P, U}$ is unbiased.

\item \textbf{Set Intersection.} To our knowledge, prior work has not considered extending $\romIAdapt$
to handle set operations other than set-union on constituent streams. However, it is possible to tweak $\romIAdapt$
in a manner similar to $\rommultiKMV$ to handle these operations.
\end{itemize}

\subsection{Other Related Work}
\label{sec:related}
Estimating the number of distinct values for data streams is a well studied problem. 
The problem of estimating result sizes of set expressions over multiple streams 
was concretely formulated by Ganguly et al. \cite{ganguly2003processing}. 
Motivated by the question of handling streams containing both insertions and deletions, their construction
involves a 2-level hash function that essentially stores a set of counters for each bit-position of an HLL-type hash, 
and hence is inherently more resource intensive, both in terms of the space and update times. 

K'th Minimum Value sketches were introduced by Bar-Yossef et al. \cite{bar2002counting}, and developed into
an unbiased scheme that handles set expressions by Beyer et al. \cite{beyer2009distinct}.
Our own scheme is closely related to the schemes
proposed and analyzed in Cohen and Kaplan~\cite{cohen2009leveraging}, 
and in Gibbons and Tirthapura~\cite{gibbons2001estimating}.
Chen, Cao and Bu~\cite{chen2007simple} propose a somewhat different
scheme for estimating unique counts with set expressions
that is based on a data-structure related to the ``probabilistic counting'' sketches of \cite{flajolet1985probabilistic},
and also to the multi-bucket KMV sketches of \cite{giroire2009order} (with $K=1$).
However, the guarantees proved by \cite{chen2007simple} are asymptotic in nature,
and their system's union sketches are the same size as base sketches, and therefore do not provide the increased
accuracy that is possible with a ``growing'' union rule as in \cite{cohen2009leveraging}, in \cite{gibbons2001estimating},
and in this paper's scheme.

Bottom-k sketches~\cite{cohen2007summarizing,cohen2009leveraging} are a weighted generalization
of KMV that provides unbiased estimates of the weights of arbitrary subpopulations of identifiers.
They have small errors even under 2-independent hashing~\cite{thorup2013bottomk}. A closely 
related method for estimating subpopulation weights is priority sampling~\cite{DuffieldLT07}.
Although this paper's Theta-Sketch Framework offers a broad generalization of KMV, 
it is not clear that it can support the entire generality of bottom-k sketches for weighted sets. 


This paper's ``Alpha Algorithm'' is inspired by the elegant {\em Approximate Counting} 
method of Morris~\cite{morris1978counting}, that has previously been
applied 
to the estimation of the frequency moments $F_p$, for $p \ge 1$. 
By contrast, {\em our} task is to estimate $\distinct_P$. The Alpha Algorithm is
able to do this because its Approximate Counting process is tightly
interleaved with another process that removes duplicates from the
input stream while maintaining a small memory footprint by using
feedback from the approximate counter.

Kane et al. \cite{kane2010optimal} gave a streaming algorithm for the \DistinctElements\ problem
that outputs a $(1+\epsilon)$-approximation with constant probability, using
$\Theta(\epsilon^{-2} + \log(n))$ bits of space. 
This improves over the bit-complexity of HLL by roughly a $\log\log n$ factor (and avoids the assumption of truly random hash functions). 
Like HLL, it is not known how to extend the algorithm to handle $\distinctsub_P$ queries for non-trivial properties $P$, and the algorithm does not appear to 
have been implemented \cite{heule2013hll}. 

Tirthapura and Woodruff \cite{finalcite} give sketching algorithms for estimating $\distinctsub_P$ queries
for a special class of properties $P$. Specifically, they consider streams that contain tuples of the form $(x,y)$, where $y$
is a numerical parameter, and the subpopulation $P$ is specified via a lower or upper bound on $y$. 

In very recent work, Cohen \cite{cohennew} and Ting \cite{ting} have proposed new estimators for \DistinctElements\ (called "Historical Inverse Probabililty" (HIP) estimators in \cite{cohennew}).
Any sketch which is generated by hashing of each element in the data stream and
is not affected by duplicate elements (such as HLL, KMV, Adaptive Sampling, and our Alpha Algorithm) has a corresponding HIP estimator, and \cite{cohennew, ting} show that the HIP estimator reduces the
variance of the original sketching algorithm by a factor of 2. However, HIP estimators, in general, can only be computed when processing the stream, and this applies
in particular to the HIP estimators of KMV and Adaptive Sampling. Hence, they do not satisfy the mergeablity properties necessary to apply to multi-stream settings.




\appendix
\section{Proof of Theorem \ref{thm:variance}}
\label{app:variance}
\subsection{Proof Overview} 
The proof introduces the notion of the \emph{fix-all-but-two projection}
of a threshold choosing function $T$. We then introduce 
a new condition on TCF's that
we call $2$-Goodness (cf. Appendix \ref{sec:twogood}). On its face, $2$-Goodness may appear to be a stronger requirement than $1$-Goodness. However,
we show in Section \ref{1implies2} that this is not the case: $1$-Goodness in fact implies $2$-Goodness.\footnote{In fact, the two properties can
be shown to be equivalent. We omit the reverse implication, since we will not require it to establish our variance bounds.}
We show in Appendix \ref{sec:zerocovariance} that $2$-Goodness implies that ``per-identifier estimates'' output by the Theta-Sketch Framework
are uncorrelated. Finally, in Section \ref{sec:finalvariance}, we use this result to complete the proof of Theorem \ref{thm:variance}.

\subsection{Definition of Fix-All-But-Two Projections and $2$-Goodness}
\label{sec:twogood} 
We begin by defining the Fix-All-But-Two Projection of a TCF.

\begin{definition}\label{def:fabt-projection}
Let $T$ be a threshold choosing function and fix a stream $A$.
Let $\ell_1 \neq \ell_2$ be two of the $n_A$ unique identifiers in $A$. Let $\xnmltwo$ be a fixed assignment of 
hash values to all unique identifiers in $A$ {\em except} for $\ell_1$ and $\ell_2$. 
Then the fix-all-but-two projection $T_{\ell_1, \ell_2}[\xnmltwo](x_{\ell_1}, x_{\ell_2}) : [0,1) \times [0,1) \rightarrow (0,1]$ of $T$
is the function that maps values of $(x_{\ell_1}, x_{\ell_2})$ to theta-sketch thresholds via the definition
$T_{\ell_1, \ell_2}[\xnmltwo](x_{\ell_1, \ell_2}) = T(X^{n_A}),$ where $X^{n_A}$ is the obvious combination of $\xnmltwo$, $x_{\ell_1}$, and $x_{\ell_2}$.
\end{definition}

Next, we define the notion of $2$-Goodness for bivariate functions. 
\begin{definition}\label{def:bivariate2goodness}
Let $f(x,y):[0,1) \times [0,1) \rightarrow (0,1]$ be a bivariate function. We say that $f$ satisfies $2$-Goodness if
there exists an $F \in (0, 1]$ such that

\begin{itemize}
\item $\max(x,y)   < F \Rightarrow f(x,y) = F$.
\item $\max(x,y) \ge F \Rightarrow f(x,y) \le \max(x,y)$.
\end{itemize}
\end{definition}

Finally we are ready to define $2$-Goodness for TCF's.

\begin{condition}\label{key-condition2}
A threshold choosing function $T(X^{n_A})$ satisfies $2$-Goodness iff for every stream $A$
containing $n_A$ unique identifiers, every pair of identifiers $\ell_1, \ell_2 \in A$, and every fixed 
assignment $X^{n_A}_{-\ell_1, -\ell_2}$ of hash values to the identifiers in $A \!\setminus\! \{\ell_1, \ell_2\}$,
the fix-all-but-two projection $\txntwo$ satisfies Definition~\ref{def:bivariate2goodness}.
\end{condition}

\subsection{$1$-Goodness Implies $2$-Goodness}
\label{1implies2}
We are ready to show the (arguably surprising) result that if $T$ satisfies $1$-Goodness, then it also satisfies $2$-Goodness. 

\begin{theorem} \label{thm:1implies2}
Let $T$ be a threshold choosing function that satisfies $1$-Goodness. Then $T$ also satisfies $2$-Goodness.
\end{theorem}
\begin{proof}
Let $T_{\ell_1, \ell_2}[\xnmltwo]$ be any fix-all-but-two projection of $T$. Notice that
for any $y' \in [0,1)$, $f(x) := T_{\ell_1, \ell_2}[\xnmltwo](x, y')$ is a fix-all-but-one projection of $T$.
Similarly for any $x' \in [0,1)$, $g(y) := T_{\ell_1, \ell_2}[\xnmltwo](x', y)$ is a fix-all-but-one-projection of $T$. Hence, $1$-Goodness of $T$
implies the following conditions hold:

\medskip \noindent \textbf{Property 1.} For each $y' \in [0,1)$, there exists a $G^{y'} \in (0, 1]$ such that:
\begin{itemize}
\item $x   < G^{y'} \Rightarrow T_{\ell_1, \ell_2}[\xnmltwo](x,y') = G^{y'}$.
\item $x \ge G^{y'} \Rightarrow T_{\ell_1, \ell_2}[\xnmltwo](x,y') \le x$.
\end{itemize}

\medskip \noindent \textbf{Property 2.} For each $x' \in [0,1)$, there exists a $H^{x'} \in (0, 1]$ such that:
\begin{itemize}
\item $y   < H^{x'} \Rightarrow T_{\ell_1, \ell_2}[\xnmltwo](x',y) = H^{x'}$.
\item $y \ge H^{x'} \Rightarrow T_{\ell_1, \ell_2}[\xnmltwo](x',y) \le y$.
\end{itemize}

To establish that $T$ satisfies $2$-Goodness, we want to prove that  there exists an $F \in (0, 1]$ such that

\begin{itemize}
\item $\max(x,y)   < F \Rightarrow T_{\ell_1, \ell_2}[\xnmltwo](x,y) = F$.
\item $\max(x,y) \ge F \Rightarrow T_{\ell_1, \ell_2}[\xnmltwo](x,y) \le \max(x,y)$.
\end{itemize}

\noindent We will break the proof down into two lemmas.

\begin{lemma}\label{lemma1}
There exists an $F \in (0, 1]$ such that
$\max(x,y)   < F \Rightarrow T_{\ell_1, \ell_2}[\xnmltwo](x,y) = F$.
\end{lemma}
\begin{proof}

\noindent By Property 1 above, there exists a $G^0 \in (0, 1]$ such that 
\begin{equation}\label{eqn1}
x   < G^{0} \Rightarrow T_{\ell_1, \ell_2}[\xnmltwo](x,0) = G^{0}.
\end{equation}

\noindent Now consider any $x$ in $[0,G^0)$.
By Property 2 above, there exists a $H^{x} \in (0, 1]$ such that:
\begin{equation}\label{eqn2}
y   < H^{x} \Rightarrow T_{\ell_1, \ell_2}[\xnmltwo](x,y) = H^{x}.
\end{equation}

\noindent Plugging $y=0$ into Equation~\eqref{eqn2} gives
$T_{\ell_1, \ell_2}[\xnmltwo](x,0) = H^{x}$, while Equation~\eqref{eqn1} guarantees that
$T_{\ell_1, \ell_2}[\xnmltwo](x,0) = G^0$, so $H^{x} = G^0$. Substituting $G^0$
into Equation~\eqref{eqn2} yields 
\begin{equation}\label{eqn3}
y   < G^0 \Rightarrow T_{\ell_1, \ell_2}[\xnmltwo](x,y) = G^0.
\end{equation}

\noindent Because $x$ was any value in the interval $[0,G^0$), the lemma is proved with $F = G^0$.
\end{proof}

\begin{lemma}\label{lemma2}
The threshold $F$ whose existence was proved in Lemma~\ref{lemma1} also has the
property that if $\max(x,y) \ge F$, then $T_{\ell_1, \ell_2}[\xnmltwo](x,y) \le \max(x,y)$.

\end{lemma}
\begin{proof}
\noindent We start by assuming that $\max(x,y) \ge F$, so at least one of the following must be true: 
($x \ge F$) or ($y \ge F$). 
Without loss of generality we will assume that $x \ge F$. 
\noindent By Property 2 above, there exists an $H^{x} \in (0, 1]$ such that
\begin{itemize}
\item $y   < H^{x} \Rightarrow T_{\ell_1, \ell_2}[\xnmltwo](x,y) = H^{x}$.
\item $y \ge H^{x} \Rightarrow T_{\ell_1, \ell_2}[\xnmltwo](x,y) \le y$.
\end{itemize}

\noindent Our proof will have two cases, determined by whether 
$y < H^x$ or $y \ge H^x$.

\noindent First case: $y < H^x$. In this case, because $y < H^x$, $T_{\ell_1, \ell_2}[\xnmltwo](x,y) = H^x$. Also, $T_{\ell_1, \ell_2}[\xnmltwo](x,0) = H^x$.
But $x \ge F = G^0$, so $T_{\ell_1, \ell_2}[\xnmltwo](x,0) \le x$. Putting this all together gives:
\begin{equation}
T_{\ell_1, \ell_2}[\xnmltwo](x,y) = H^x = T_{\ell_1, \ell_2}[\xnmltwo](x,0) \le x \le \max(x,y).
\end{equation}

\noindent Second case: $y \ge H^x$. In this case,
because $y \ge H^x$,
\begin{equation}
T_{\ell_1, \ell_2}[\xnmltwo](x,y) \le y \le \max(x,y).
\end{equation}
\end{proof}
\end{proof}

\begin{figure*}
\begin{center}
\includegraphics[width=0.5\linewidth]{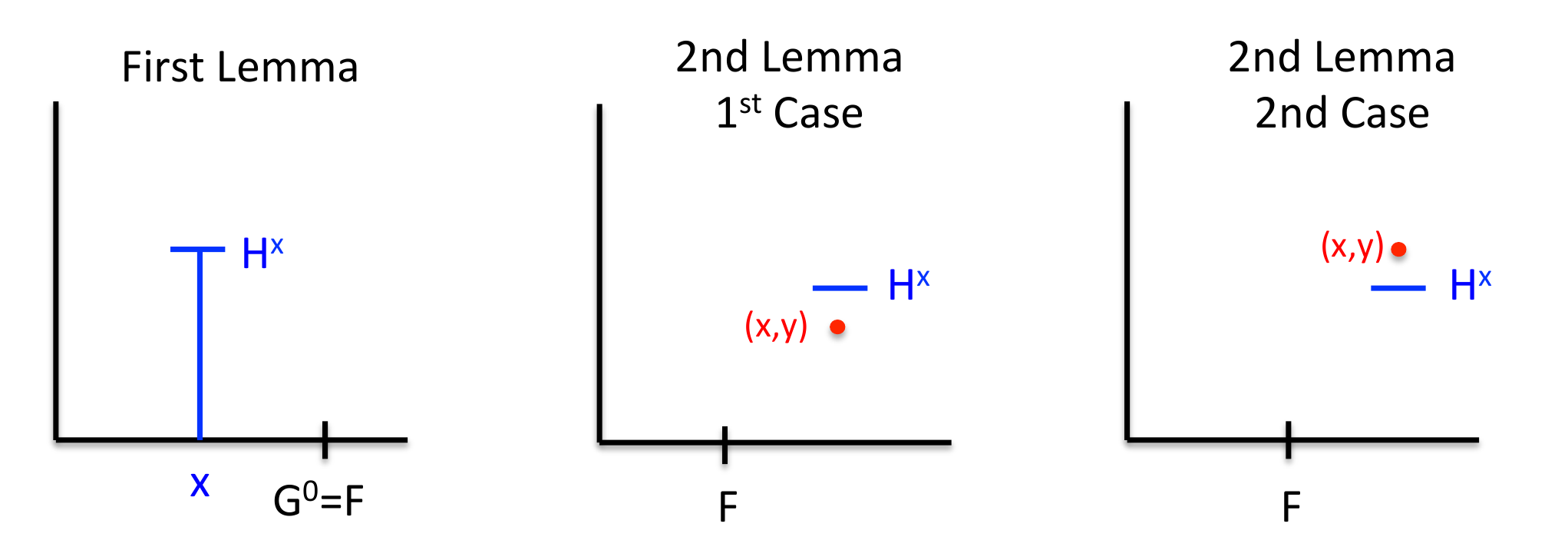}
\end{center}
\caption{Some diagrams for Lemmas \ref{lemma1} and \ref{lemma2}}
\end{figure*}

\noindent \textbf{$1$-Goodness Implies Per-Identifier Estimates Are Uncorrelated}
 \label{sec:zerocovariance}

\begin{lemma} \label{lemma:zerocovariances} Fix any stream $A$, threshold choosing function $G$, and pair $\ell_1 \neq \ell_2$ in $A$. Define the ``per-identifier estimates'' $V_{\ell_1}$ and $V_{\ell_2}$ as in Equation~\eqref{define-v-ell}. Then
if $T$ satisfies $1$-Goodness, the covariance of $V_{\ell_1}$ and $V_{\ell_2}$ is 0. In symbols,
$$\sigma(V_{\ell_1}, V_{\ell_2}) = E_{\xn}(V_{\ell_1} \cdot V_{\ell_2}) -E_{\xn}(V_{\ell_1}) \cdot E_{\xn}(V_{\ell_2})=0.$$
\end{lemma}
\begin{proof}
Because $T$ satisfies $1$-Goodness, it also satisfies $2$-Goodness (cf. Theorem \ref{thm:1implies2}), and hence there exists a threshold $F(\xnmltwo)$ for which it is a straightforward 
exercise to verify that:

\begin{equation}
V_{\ell_1}(\xn) \cdot V_{\ell_2}(\xn) 
= \left\{
\begin{array}{l}
1 / F(\xnmltwo)^2
\;\mathrm{if}\; \max (x_{\ell_1}, x_{\ell_2}) < F(\xnmltwo) \\
0 \;\mathrm{otherwise}.
\end{array}\right.
\end{equation}

\noindent Now, conditioning on $\xnmltwo$ and taking the expectation with respect to pairs ($x_{\ell_1}$, $x_{\ell_2}$):

\begin{equation}
E(V_{\ell_1} \cdot V_{\ell_2} | \xnmltwo) = \int_0^1 \int_0^1 
V_{\ell_1}(\xn) V_{\ell_2}(\xn) 
dx_{\ell_1} dx_{\ell_2}
= F(\xnmltwo)^2 \cdot \frac{1}{F(\xnmltwo)^2} = 1.
\end{equation}

\noindent Since 
$E(V_{\ell_1} V_{\ell_2} | \xnmltwo) = 1$
when conditioned on each 
$\xnmltwo$, 
we also have
$E(V_{\ell_1} V_{\ell_2}) = 1$
when the expectation is taken over all $\xn$.
Meanwhile, since $T$ satisfies $1$-Goodness, 
$E(V_{\ell_1}) = E(V_{\ell_2})=1$ (cf. Theorem \ref{thm:unbiased}).
Hence, $\sigma(V_{\ell_1}, V_{\ell_2}) = 0$.
\end{proof}

As a corollary of Lemma \ref{lemma:zerocovariances}, we obtain the following result, establishing that the variance 
of $\romTS_{P, A}$ is equal to the sum of the variances of the per-identifier estimates for all identifiers in $A$ satisfying property $P$.

\begin{lemma} \label{the_corollary} Suppose that $T$ satisfies $1$-Goodness. Fix any stream $A$, and let $\romTS_{P, A}$
denote the estimate for $n_{P, A}$ obtained by running $\samp[T]()$ on $A$ and feeding the resulting theta-sketch into $\estonsub$().
Then
$$\sigma^2(\romTS_{P,A}) = \sum_{\ell \in A \colon P(\ell) = 1} \sigma^2(V_\ell).$$
\end{lemma}
\begin{proof}
Note that $\romTS_{P,A} = \sum_{\ell \in A \colon P(\ell)=1} V_\ell$. The claim then follows from Lemma \ref{lemma:zerocovariances} combined with the fact
that the variance of the sum of random variables equals the sum of the variances, provided that the variables appearing in the sum are uncorrelated. 

\end{proof}

\subsection{Completing the Proof of Theorem \ref{thm:variance}}
\label{sec:twoimpliesvariance}
\label{sec:finalvariance} 
\begin{proof}
For every $\ell$ that appears in the concatenated stream $A^*$, and for all $\xnstar$,
we define the ``per-identifier estimate'' $V_\ell(\xnstar)$ as in Equation~\eqref{define-v-ell} with $A=A^*$, and relate it to the threshold
$F_{\ell}(\xnmlstar)$ as in Equation~\eqref{v-in-terms-of-f}, also with $A=A^*$. It is then straightforward to verify that

\begin{equation}\label{var-v-from-f}
\sigma^2(V_\ell | \xnmlstar) = 1 / F_\ell(\xnmlstar) - 1.
\end{equation}

\noindent Let $T'$ be the TCF that was (implicitly) used to construct $(\theta^U,S^U)$ from the $m$ sketches of the individual streams $A_j$.
By Theorem \ref{thm:preserved}, $T'$ satisfies $1$-Goodness, so let $F'_{\ell}(\xnmlstar)$ denote the corresponding threshold value for $T'$ as in Equation~\eqref{v-in-terms-of-f}. 
We claim that $T'$ satisfies the following property:
\begin{equation}\label{importantprop} \text{For all identifiers } \ell \in [n] \text{ and for all } \xnstar, F'_\ell(\xnmlstar) \ge F_\ell(\xnmlstar). \end{equation} 

\noindent \textbf{Finishing the proof, assuming $T'$ satisfies Property \ref{importantprop}.}
By Equation~\eqref{var-v-from-f}:

\begin{equation}
\sigma^2(V'_\ell | \xnmlstar) \le \sigma^2(V_\ell | \xnmlstar).
\end{equation}

\noindent Because this inequality holds for every specific $\xnmlstar$, it also holds for any convex combination over 
$\xnmlstar$'s, so 

$$\sigma^2(V'_\ell) \le \sigma^2(V_\ell).$$

Combining this with Lemma \ref{the_corollary},
we conclude that 
$$\sigma^2(\hat{n}^U_{P,U}) \;=\; \sum_{\ell \in A : P(\ell) = 1} \sigma^2(V'_\ell) \;\le\;
\sum_{\ell \in A : P(\ell) = 1} \sigma^2(V_\ell) \;=\; \sigma^2(\hat{n}^{A^*}_{P,A^*}).$$

\noindent \textbf{Proving that $T'$ satisfies Property \ref{importantprop}.}
Fix any hash function $h$, which determines $\xnu$, and also fixes hashed versions of
the streams $A_1, \dots, A_m$ and $A^*$. We will overload the symbols $A_j$ and $A^*$ to denote these hashed streams
as well as the original streams.
We need to prove that $F'_{\ell}(\xnuml) \ge F_{\ell}(A^*_{-\ell})$. 
This can be done in three steps. First, from the proof of Theorem~\ref{thm-union-preserves-condition}
we know that there exists a $j$ such that $F'_{\ell}(\xnuml) = F_{\ell}(A_{j, -\ell})$.
Second, because $T$ satisfies $1$-Goodness, 
$F_{\ell}(A_{j, -\ell}) = T(Z(A_j,\ell))$ and $F_{\ell}(A^*_{-\ell}) = T(Z(A^*,\ell))$, where $Z$ is a function that 
makes a copy of a hashed stream in which $h(\ell)$ has been artificially set to zero.
Third, $Z(A^*,\ell))$ can be rewritten as the concatenation of 3 streams as follows: $B_0 \circ Z(A_j,\ell)
\circ B_2$, where $B_0 = Z(A_1, \ell) \circ Z(A_2, \ell) \circ \dots, Z(A_{j-1}, \ell)$, and $B_2 = Z(A_{j+1}, \ell) \circ \dots \circ Z(A_m, \ell)$. 
Because $T$ was assumed to satisfy the monotonicity condition, Condition \ref{mildcondition}, we then have
\begin{equation}
F'_{\ell}(\xnuml) = 
T(Z(A_j,\ell)) \ge 
T(B_1 \circ Z(A_j,\ell) \circ B_3) =
T(Z(A^*,\ell)) =
F_{\ell}(A^*_{-\ell}). 
\end{equation}

\end{proof}

\section{Details of the Analysis of the Alpha Algorithm for Single Streams}

\subsection{Proof of Theorem~\ref{space-theorem}}\label{appendix-proof-of-space-theorem}

Let $S$ be the set produced by 
Line~\ref{code:framework-line-3}
of
Algorithm~\ref{code:framework}
when AlphaTCF is plugged
into the Theta Sketch Framework, 
and let $\mathcal{S}$ be the random variable corresponding to $|S|$.
In this section we compute $E(\mathcal{S})$ and bound $\sigma^2(\mathcal{S})$. 

We prove the top-level theorem using a lemma. The proofs of the theorem and lemma
both involve two levels of conditioning. First we condition on the value $\mathcal{I}$ of $i$ when 
Line \ref{code:alpha-tcf-line-15} of 
Algorithm~\ref{code:alpha-tcf} is reached.
Then we further condition on $\mathcal{J}^+$, which we define to be the
particular set of $i$ stream positions on which increments occurred in 
Line~\ref{code:alpha-tcf-line-10} of Algorithm~\ref{code:alpha-tcf}.

\subparagraph*{Restatement of Theorem~\ref{space-theorem}.}
\begin{align}
\ept(\mathcal{S}) = & \; k. \\
\sigma^2(\mathcal{S}) < & \; \frac{k}{2} + \frac{1}{4}.
\end{align}


\begin{proof}
Using standard laws of probability, we perform the following decompositions:
\begin{align}
E(\mathcal{S}) =        & \sum_i \Pr (\mathcal{I}=i) \ept(\mathcal{S} | \mathcal{I}=i) \label{s-eqn-1} \\
\ept(\mathcal{S}| \mathcal{I}=i) = & \sum_{J} \Pr (\mathcal{J}^+=J | \mathcal{I}=i) \ept(\mathcal{S} | \mathcal{I}=i, \mathcal{J}^+ = J) \label{s-eqn-2} \\
E(\mathcal{S}^2) =      & \sum_i \Pr (\mathcal{I}=i) \ept (\mathcal{S}^2 | \mathcal{I}=i) \label {s-eqn-3} \\
\ept_i(\mathcal{S}^2|\mathcal{I}=i) = & \sum_{J} \Pr (\mathcal{J}^+=J | \mathcal{I}=i) \ept (\mathcal{S}^2 | \mathcal{I}=i, \mathcal{J}^+ = J) \label{s-eqn-4}
\end{align}
In Lemma~\ref{space-lemma}, we prove that for all $i$ and $J$,
\begin{align}
\ept(\mathcal{S} | \mathcal{I}=i, \mathcal{J}^+ = J) = & k.
\end{align}
Because this answer does not depend on $J$, the RHS of 
Equation~\eqref{s-eqn-2} is a convex combination of equal values, so
\begin{align*}
\ept(\mathcal{S}|\mathcal{I}=i) = & k.
\end{align*}
Because this answer does not depend on $i$, the RHS of 
Equation~\eqref{s-eqn-1} is a convex combination of equal values, so
\begin{align*}
E(\mathcal{S}) = & k. \\
\end{align*}
In Lemma~\ref{space-lemma}, we also prove that
\begin{align*}
\ept(\mathcal{S}^2 | \mathcal{I}=i, \mathcal{J}^+=J) = & k^2 + \frac{\alpha - \alpha^{2i+1}}{1-\alpha^2}.
\end{align*}
Because this answer does not depend on $J$, the RHS of 
Equation~\eqref{s-eqn-4} is a convex combination of equal values, so
\begin{align}
\ept(\mathcal{S}^2|\mathcal{I}=i) = & k^2 + \frac{\alpha - \alpha^{2i+1}}{1-\alpha^2}. \label{s-eqn-5}
\end{align}
This answer {\em does} depend on $i$, so we cannot use the exact same argument for a fourth time.
However, with a little bit of algebra, one can go from Equation~\eqref{s-eqn-5} to the inequality
\begin{align*}
\ept(\mathcal{S}^2|\mathcal{I}=i) < & k^2 + \frac{k}{2} + \frac{1}{4},
\end{align*}
whose RHS does not depend on $i$. Then the RHS of Equation~\eqref{s-eqn-3}
is a convex combination of values that are all less than 
$k^2 + \frac{k}{2} + \frac{1}{4}$. So:
\begin{align*}
\ept(\mathcal{S}^2) < & k^2 + \frac{k}{2} + \frac{1}{4}. \\
\sigma^2(\mathcal{S}) < & (k^2 + \frac{k}{2} + \frac{1}{4}) - k^2 = \frac{k}{2} + \frac{1}{4}.
\end{align*}
\end{proof}

\noindent Now we will prove the lemma that was used above:

\begin{lemma}\label{space-lemma}
For all $i$ and $J$,
\begin{align}
\ept (\mathcal{S} | \mathcal{I}=i, \mathcal{J}^+=J) = & \; k, \label{space-lemma-line-1} \\
\ept (\mathcal{S}^2 | \mathcal{I}=i, \mathcal{J}^+=J) = & \; k^2 + \frac{\alpha - \alpha^{2i+1}}{1-\alpha^2}.
\end{align}
\end{lemma}
\begin{proof}
Because Line \ref{code:alpha-tcf-line-9} 
of Algorithm~\ref{code:alpha-tcf} causes the algorithm to ignore duplicate labels,
it will suffice to analyze a stream $A$ of length $n$ that doesn't contain any
duplicates.
Let $h$ be a hash function that is
chosen randomly. Let $\{ X_p | 1 \le p \le n \}$ be a sequence of $n$ {\em iid}
random variables, one per stream position, each drawn from the distribution Uniform(0,1). 
Let $X^{n_A}$ be the cross product of the $X_p$'s; this random variable is our model of $h(A)$.
$\mathcal{I}$ is distributed as a random variable generated by first choosing a random $X^{n_A}$, then running Algorithm~\ref{code:alpha-tcf} on $X^{n_A}$,
and then setting $\mathcal{I}$ to be the value of the program variable $i$ when Line~\ref{code:alpha-tcf-line-15} is reached.
Define a set of $n$ Bernoulli random variables $S_p$, one
per stream position, derived from the variable $\mathcal{I}$ and the variables
$X_p$ by the rule $S_p = 1$ iff $X_p < \alpha^i$.
Note that the $S_p$'s are {\em not} independent of each other. However, as we will see, they become 
independent after conditioning on the event ($\mathcal{I}=i$ and $\mathcal{J}^+=J$).

Now we will describe the effect of conditioning on both $\mathcal{I}=i$ and $\mathcal{J}^+=J$,
by first describing how the original variables $X_p$ are transformed
into modified variables $Y_p$ that are drawn from specific
subintervals of $(0,1)$.  We will then introduce new Bernoulli
variables $S'_p$ defined by the rule $S'_p = 1$ iff $Y_p <
\alpha^i$, and finally compute the expected value and variance of
$(\mathcal{S}|\mathcal{I}=i,\mathcal{J}^+=J) = \sum_p S'_p$.

We are fixing a specific value $i$ of $\mathcal{I}$ and $J$ of $\mathcal{J}^+$; the latter is a size-$i$ subset of the set of 
$n-k$ non-initial stream positions $\{k+1, k+2, \ldots, n-1, n\}$.
The set of $n-k-i$ non-initial stream positions that are not in $J$ will be 
referred to as $J^-$.
Let $f(p,J)$ be the function that maps any non-initial position
$p$ 
to the number of non-initial positions before $p$ that are members of $J$.
We note that for $p \in J$, $f(p,J) \in \{0,1,\ldots,i-1\}$, and the
mapping is one-to-one.
For $p \in J^-$, $f(p,J) \in \{0,1,\ldots,i\}$, and the mapping is not necessarily one-to-one.

Now we are ready to characterize the $Y_p$'s and the $S'_p$'s. 

First let $p$ be one of the $k$ initial positions in the stream.
In this case, conditioning on $\mathcal{I}=i$ and $\mathcal{J}^+=J$ does not tell us anything about the value of $X_p$,
so $Y_p$ is drawn from the full interval $(0,1)$, so $\Pr(Y_p < \alpha^i) = \alpha^i$;
$E(\mathcal{S}'_p) = \alpha^i$, and $\sigma^2(S'_p) = \alpha^i (1 - \alpha^i)$.

Next, let $p$ be one of the $n-k-i$ positions in $J^-$. 
For this position, the test in 
Line \ref{code:alpha-tcf-line-8} of Algorithm \ref{code:alpha-tcf}
failed, so we know that $X_p \ge \alpha^{f(p,J)}$,
so $Y_p$ is drawn uniformly from the interval $[\alpha^{f(p,J)},1)$. Because $p \in J^-$,
$f(p,J) \le i$, so $\alpha^i \le \alpha^{f(p,J)} \le Y_p$, 
so $\Pr(Y_p < \alpha^i) = 0$, $E(S'_p) = 0$, and $\sigma^2(S'_p) = 0$.

Finally, let $p$ be one of the $i$ positions that are in $J$. For this position,
the test in 
Line \ref{code:alpha-tcf-line-8} of Algorithm \ref{code:alpha-tcf}
succeeded, so we know that $X_p < \alpha^{f(p,J)}$,
so $Y_p$ is drawn uniformly from the interval $(0,\alpha^{f(p,J)})$, so
$\Pr(Y_p < \alpha^i) = \alpha^i/\alpha^{f(p,J)} = \alpha^{i-f(p,J)}$.
Now, because $f(p,J)$ assumes each value in $\{0,1,\ldots,i-1\}$ as
$p$ is varied over the contents of $J$, $E(S'_p) = \Pr(Y_p < \alpha^i) =
\alpha^{i-f(p,J)}$ assumes each value in 
$\{\alpha^i, \alpha^{i-1}, \ldots, \alpha^1\}$.
Similarly, $\sigma^2(S'_p)$ assumes each value in
$\{
\alpha^i (1\!- \!\alpha^i),
\ldots
\alpha^1 (1\!- \!\alpha^1)
\}$. 
Note that the above analysis implies that the $S_p$'s are independent of each other after conditioning on the 
event ($\mathcal{I}=i$ and $\mathcal{J}^+=J$).

Putting together all of the above, and remembering that the random variables
$S'_p$ are independent due to the conditioning on $\mathcal{I}=i$ and $\mathcal{J}^+=J$:
\begin{align*}
E(\mathcal{S}|\mathcal{I}=i,\mathcal{J}^+=J) = & k \cdot \alpha^i + (n-k-i) \cdot 0 + \sum_{j=1}^{i} \alpha^j = k.
\end{align*}
Here, we have used the fact that $\alpha = \frac{k}{k+1}$.
In addition:
\begin{align*}
\sigma^2(\mathcal{S}| \mathcal{I}=i,\mathcal{J}^+=J) = & k \alpha^i \cdot (1 \!-\! \alpha^i) + (n\!-\!k\!-\!i) \cdot 0 + \sum_{j=1}^{i} \alpha^j (1 \!-\!\alpha^j) \\
             = & \frac{\alpha - \alpha^{2i+1}}{1 - \alpha^2}, \\
E(\mathcal{S}^2|\mathcal{I}=i,\mathcal{J}^+ = J) 
= & k^2 + \frac{\alpha - \alpha^{2i+1}}{1 - \alpha^2}
\end{align*}
\end{proof}

\subsection{Proof of Theorem~\ref{thm:alpha-basic-variance}}\label{appendix-proof-of-alpha-basic-variance}


\noindent {\bf Preliminaries and Notation:} Because Line~\ref{code:alpha-tcf-line-9} 
of Algorithm~\ref{code:alpha-tcf} causes the algorithm to ignore duplicate labels,
it will suffice to analyze streams that do not contain any duplicates.
$\mathcal{S}$ and $\mathcal{I}$ are random variables giving
the final values of $|S|$ and $i$ when Line~\ref{code:alpha-tcf-line-15} of Algorithm
\ref{code:alpha-tcf} is
reached. Let $\mathcal{Z} = \mathcal{S} / \alpha^{\mathcal{I}}$ denote the
random variable for the estimate produced by the Theta Sketch
framework when the Alpha Algorithm's TCF is used. 
It will be convenient to introduce a new variable $u = n_A - k$
representing the number of stream items that are processed by the 
Alpha Algorithm {\em after} $k$ initial items have been processed to
initialize the set $S$.  Recall that $\alpha = k/(k+1)$. 





\subparagraph*{Restatement of Theorem~\ref{thm:alpha-basic-variance}.}
\[
\sigma^2(\mathcal{Z}) = \frac{(2k+1)n_A^2 - (k^2+k)(2n_A-1)- n_A}{2k^2} < \frac{n_A^2}{k - \frac{1}{2}}.
\]
\begin{proof}
Due the 1-goodness of the Alpha Algorithm's TCF, we already 
know that $\mathcal{Z}$ is an unbiased estimator for $n_A$, i.e., that $E(\mathcal{Z}) = n_A = k+u$. Hence,
\begin{align*}
\sigma^2(\mathcal{Z}) 
= & E(\mathcal{Z}^2) - E^2(\mathcal{Z}) \\
= & E(\mathcal{Z}^2) - (k+u)^2.
\end{align*}

Lemma~\ref{lemma:e-z-sq} (stated and proved below) gives a formula for $E(\mathcal{Z}^2)$. This allows us to complete the 
the analysis of $\sigma^2(\mathcal{Z})$ as follows:
\begin{align*}
             E(\mathcal{Z}^2) - (k+u)^2         = & \frac{k^2u + ku^2 + u(u\!-\!1)/2}{k^2} \\
                                = & \frac{(2k+1)n_A^2 - (k^2+k)(2n_A-1)- n_A}{2k^2} \\
                                          < & \frac{n_A^2}{k - \frac{1}{2}}.
\end{align*}
\end{proof}

\begin{lemma}\label{lemma:e-z-sq}
\begin{align*}
E(\mathcal{Z}^2) = & \frac{k^2u + ku^2 + u(u\!-\!1)/2 + k^4 + 2k^3u + k^2u^2}{k^2}. \\
\end{align*}
\end{lemma}
\begin{proof}
\begin{align*}
E(\mathcal{Z}^2) 
= & \sum_{i=0}^u \sum_{s=0}^{k+i} \left( \frac{s}{\alpha^i} \right)^2 
    \mathrm{Pr}(\mathcal{S}\!=\!s| \mathcal{I}=i) \; \mathrm{Pr}(\mathcal{I}\!=\!i;u) \\
= & \sum_{i=0}^u \frac{1}{\alpha^{2i}} 
\mathrm{Pr}(\mathcal{I}\!=\!i;u)
\sum_{s=0}^{k+i} s^2 
\mathrm{Pr}(\mathcal{S}\!=\!s| \mathcal{I}=i) \\
= & \sum_{i=0}^u \frac{1}{\alpha^{2i}} 
\mathrm{Pr}(\mathcal{I}\!=\!i;u)
E(\mathcal{S}^2| \mathcal{I}=i) \\
\end{align*}

Above, we use the somewhat onerous notation $\mathrm{Pr}(\mathcal{I}\!=\!i;u)$ to emphasize that the distribution
of $\mathcal{I}$ depends on the fixed quantity $u$ (i.e., on the number of distinct elements in the stream, minus $k$). Making this dependence explicit
will be useful later, when we analyze this distribution by establishing recurrences involving $u$.

A formula for $E(\mathcal{S}^2|\mathcal{I}=i)$ appeared in Equation~\eqref{s-eqn-5}. 
Substituting this formula and continuing:

\begin{align}
= & \sum_{i=0}^u \frac{1}{\alpha^{2i}} 
\mathrm{Pr}(\mathcal{I}\!=\!i;u) \notag
\left(  \frac{\alpha-\alpha^{2i+1}}{1-\alpha^2} + k^2  \right) \\
= & \label{bigexpression} \frac{1}{1-\alpha^2} 
\left[
\alpha \sum_{i=0}^u \frac{1}{\alpha^{2i}} \mathrm{Pr}(\mathcal{I}\!=\!i;u) -
\sum_{i=0}^u \alpha \mathrm{Pr}(\mathcal{I}\!=\!i;u)
\right] 
 + k^2 \sum_{i=0}^u \frac{1}{\alpha^{2i}} \mathrm{Pr}(\mathcal{I}\!=\!i;u).
 \end{align}

Define the function

\begin{equation} \label{gdef}
g(q,k,u) = \sum_{i=0}^u \frac{1}{\alpha^{q \cdot i}} \mathrm{Pr}(\mathcal{I}\!=\!i;u)
\end{equation}

\noindent where $\mathrm{Pr}(\mathcal{I}\!=\!i;u)$ is the probability distribution 
governing the random variable $\mathcal{I}$.
In Section~\ref{gqku-section} we prove Lemma~\ref{lemma:gqku-formulas}, which includes the following 
formula for $g(2,k,u)$:

\begin{equation}
g(2,k,u) = \;\frac{k^3 + 2 k^2 u + ku^2 + u(u\!-\!1)/2}{k^3}. \label{eqn:g-2-k-u}
\end{equation}

By the definition of $g(q, k, u)$, we have that Expression \eqref{bigexpression} equals:
 \begin{align*}
& \frac{1}{1-\alpha^2} 
\left[ \alpha \cdot g(2,k,u) - \alpha \cdot 1 \right] + k^2 g(2,k,u) \\
= & \left(\frac{\alpha}{1-\alpha^2} + k^2\right) \cdot g(2,k,u) - \frac{\alpha}{1-\alpha^2} \\
= & \frac{k(2k^2+2k+1)}{2k+1}
\cdot g(2,k,u) - 
\frac{k(k+1)}{2k+1}. \\
\end{align*}

\noindent Finally, substituting Equation~\eqref{eqn:g-2-k-u}'s formula for $g(2,k,u)$ and performing
elementary algebra manipulations yields the result:

\begin{align*}
= & \frac{
\begin{array}{c}
2k^5 + 4k^4u + 2k^3u^2 + 3k^2u^2 + k^4 + 4k^3u + 2ku^2 + k^2u - ku + u(u-1)/2
\end{array}
}{k^2 (2k+1)} \\
= & \frac{k^2u + ku^2 + u(u\!-\!1)/2 + k^4 + 2k^3u + k^2u^2}{k^2}.
\end{align*}
\end{proof}


\subsection{Analysis of the function $g(q,k,u)$}\label{gqku-section}

Recall from the proof of Lemma \ref{lemma:e-z-sq} that $\mathrm{Pr}(\mathcal{I}\!=\!i ; u)$ is the probability distribution
governing the final value of the Alpha Algorithm's variable $i$. The analysis
of approximate counting in \cite{flajolet1985approximate} includes an explanation
of why the following base cases and
recurrence define the distribution $\mathrm{Pr}(\mathcal{I}\!=\!i ; u)$.

\begin{align*}
\mathrm{Pr}(\mathcal{I}\!=\!0 ; 0) = & 1 \\
\mathrm{Pr}(\mathcal{I}\!=\!i ; 0) = & 0, \quad \forall i > 0 \\
\mathrm{Pr}(\mathcal{I}\!=\!0 ; u) = & 0, \quad \forall u > 0 \\
\mathrm{Pr}(\mathcal{I}\!=\!i ; u) = & (1 - \alpha^i) \cdot \mathrm{Pr}(\mathcal{I}\!=\!i ; u\!-\!1) \; +
 \alpha^{i\!-\!1} \cdot \mathrm{Pr}(\mathcal{I}\!=\!i\!-\!1 ; u\!-\!1), \quad \forall i\!>\!0, \forall u\!>\!0
\end{align*}

Recall that $\alpha = k/(k+1)$, and define the function $g(q,k,u)$ as in Equation \eqref{gdef}:

\begin{equation*}
g(q,k,u) = \sum_{i=0}^u \frac{1}{\alpha^{q \cdot i}} \mathrm{Pr}(\mathcal{I}\!=\!i;u).
\end{equation*}

We will now prove two lemmas that partially characterize the function $g(q,k,u)$.
In Lemma~\ref{lemma:gqku-recurrences},  
we will prove that $g(q,k,u)$ satisfies a certain recurrence. In Lemma~\ref{lemma:gqku-formulas}
we will use that recurrence to prove the correctness of explicit formulas for 
$g(0,k,u)$, $g(1,k,u)$, and $g(2,k,u)$.

\begin{lemma}\label{lemma:gqku-recurrences}
$g(q,k,u)$ satisfies the following base cases and recurrence:
\begin{align}
g(0,k,u) = & \;1 \notag \\
g(q,k,0) = & \;1 \notag \\
g(q,k,u\!+\!1) = & g(q,k,u) + \left(\frac{1-\alpha^q}{\alpha^q}\right) \cdot g(q\!-\!1,k,u) \label{eqn:gqku-recurrence}
\end{align}
\end{lemma}
\begin{proof}
The base cases can be verified by inspection. The recurrence can be derived from the
recurrence for $Pr(\mathcal{I}\!=\!i ; u)$ as follows:
\begin{align*} 
  g(q,k,u\!+\!1)
= & \sum_{i=0}^{u+1} \frac{1}{\alpha^{q\cdot i}} \mathrm{Pr}(i;u+1) \\
= & \left[ \sum_{i=0}^{u} \frac{1}{\alpha^{q \cdot i}} (1-\alpha^i) \mathrm{Pr}(i;u) \right] + 0 
   + 0 + \left[ \sum_{i=1}^{u+1} \frac{1}{\alpha^{q \cdot i}} (\alpha^{i-1}) \mathrm{Pr}(i\!-\!1;u) \right].
\end{align*} 
\noindent Now we will consider the two bracketed sums. For the first one:
\begin{align} 
  & \sum_{i=0}^{u} \frac{1}{\alpha^{q \cdot i}} (1-\alpha^i) \mathrm{Pr}(i;u) \notag \\
= & \sum_{i=0}^u \frac{1}{\alpha^{q\cdot i}} \mathrm{Pr}(i;u) - \sum_{i=0}^u \frac{\alpha^i}{\alpha^{q \cdot i}} \mathrm{Pr}(i;u) \notag \\
= & g(q,k,u) - \sum_{i=0}^u \frac{1}{\alpha^{(q-1)\cdot i}} \mathrm{Pr}(i;u) \notag \\
= & g(q,k,u) - g(q-1,k,u).\label{first-bracketed-sum}
\end{align}
\noindent For the second one (performing the change of variables $j = i\!-\!1$):
\begin{align} 
  & \sum_{i=1}^{u+1} \frac{1}{\alpha^{q\cdot i}} (\alpha^{i-1}) \mathrm{Pr}(i\!-\!1;u) \notag \\
= & \sum_{j=0}^{u} \frac{1}{\alpha^{q \cdot (j+1)}} (\alpha^{j}) \mathrm{Pr}(j;u)\notag \\
= & \frac{1}{\alpha^q}\sum_{j=0}^{u} \frac{1}{\alpha^{(q-1)j}} \mathrm{Pr}(j;u) \notag \\
= & \frac{1}{\alpha^q} g(q-1,k,u).\label{second-bracketed-sum}
\end{align}
\noindent Adding (\ref{first-bracketed-sum}) and (\ref{second-bracketed-sum})
yields the claimed result:
\begin{align*} 
  & g(q,k,u) - g(q-1,k,u) + \frac{1}{\alpha^q} g(q-1,k,u) \\
= & g(q,k,u) + \left(\frac{1-\alpha^q}{\alpha^q}\right) \cdot g(q\!-\!1,k,u)
\end{align*}
\end{proof}

\begin{lemma}\label{lemma:gqku-formulas}
For all integers $u \geq 0$: 
\begin{align*}
g(0,k,u) = & \;1 \\
g(1,k,u) = & \;\frac{k+u}{k} \\
g(2,k,u) = & \;\frac{k^3 + 2 k^2 u + ku^2 + u(u\!-\!1)/2}{k^3}
\end{align*}
\end{lemma}
\begin{proof}
$g(0,k,u) =1$ can be verified by inspection.

\vspace{0.5em}
\noindent The result for $q=1$ can be proved by noting that
$\frac{k+u}{k}$ satisfies the same base cases and recurrence as $g(1,k,u)$.
The base cases can be verified by inspection. The recurrence (\ref{eqn:gqku-recurrence})
can be verified as follows:
\begin{equation*}
\frac{k+(u+1)}{k} = \frac{k+u}{k} + \frac{1-\alpha}{\alpha} \cdot 1 = \frac{k+u}{k} + \frac{1}{k}.
\end{equation*}

\vspace{0.5em}
\noindent The result for $q=2$ can be proved by noting that
$\frac{k^3 + 2 k^2 u + ku^2 + u(u\!-\!1)/2}{k^3}$ satisfies the same base cases and recurrence as $g(2,k,u)$.
The base cases can be verified by inspection. The recurrence (\ref{eqn:gqku-recurrence})
can be verified as follows:
\begin{align*}
   & g(2,k,u) + \left(\frac{1-\alpha^2}{\alpha^2}\right) \cdot g(1,k,u) \\
 = & g(2,k,u) + \frac{1}{\alpha^2} \cdot \frac{k+u}{k} - \frac{k+u}{k} \\
 = & g(2,k,u) + \frac{(k+1)^2}{k^2} \cdot \frac{k+u}{k} - \frac{k+u}{k} \\
 = & g(2,k,u) + \frac{2k^2+k+2uk+u}{k^3} \\
 = & \frac{k^3 + 2k^2u +2k^2 + ku^2 + 2uk + k + \frac{1}{2} u^2 + \frac{1}{2} u}{k^3} \\
 = & \frac{k^3 + 2k^2(u+1) + k(u+1)^2 + (u+1)u/2}{k^3} \\
 = & g(2,k,u+1).
\end{align*}
\end{proof}

\subsection{Analysis of the Alpha Algorithm's HIP Estimator: Proof of Theorem \ref{thm:hip}}
\label{app:hip}

\begin{proof}
As in Appendix~\ref{appendix-proof-of-alpha-basic-variance}, let
 $u = n_A - k$, and let $\mathcal{I}$ denote the random variable
for the final value of the Alpha algorithm's variable $i$.
\begin{align*}
 E(k/\alpha^{\mathcal{I}})  = & \sum_{i=0}^{u}\frac{k}{\alpha^i} \cdot \mathrm{Pr}(\mathcal{I}\!=\!i ; u) \\ 
  = & k \cdot g(1,k,u) \\
= & k\cdot\frac{k+u}{k} = (k+u) = n_A \\
\vspace{0.75em}
 \sigma^2(k/\alpha^{\mathcal{I}}) = & - E^2(k/\alpha^{\mathcal{I}}) + E((k/\alpha^{\mathcal{I}})^2) \\
 = & -n_A^2 + \sum_{i=0}^{u}\frac{k^2}{\alpha^{2i}} \cdot \mathrm{Pr}(\mathcal{I}\!=\!i ; u) \\ 
  = & -n_A^2 + k^2 \cdot g(2,k,u) \\
  = & \frac{u(u-1)}{2k} \\
  = & \frac{n_A^2-2n_Ak+k^2-n_A+k}{2k} \\
  < & \frac{n_A^2}{2k}
\end{align*}
\end{proof}

\end{document}